\documentclass[11pt]{article}
\usepackage[ruled,vlined,linesnumbered]{algorithm2e}
\usepackage{amsfonts}
\usepackage{amssymb}
\usepackage{amstext}
\usepackage{amsmath}
\usepackage{amsthm}
\usepackage{bbm}
\usepackage{enumitem}
\usepackage{geometry}
\usepackage{mathrsfs}
\usepackage{tikz}
\usepackage{float}
\usetikzlibrary{decorations.pathreplacing,angles,quotes}

\newtheorem{lemma}{Lemma}
\newtheorem{theorem}{Theorem}

\newtheorem{claim}{Claim}
\theoremstyle{definition}
\newtheorem{definition}{Definition}

\newcommand{\R}{\mathbb R}
\newcommand{\Z}{\mathbb Z}

\newcommand{\1}{\mathbbm 1}

\newcommand{\size}[1]{\ensuremath{\left|#1\right|}}
\newcommand{\ceil}[1]{\ensuremath{\left\lceil#1\right\rceil}}

\newcommand{\set}[1]{\ensuremath{\left\{#1\right\}}}
\newcommand{\pr}[1]{\ensuremath{\left(#1\right)}}
\newcommand{\br}[1]{\ensuremath{\left[#1\right]}}

\DeclareMathOperator*{\argmin}{arg\,min}

\DeclareMathOperator{\supp}{supp}

\tikzstyle{vertex}=[circle,draw,align=center]
\tikzstyle{weight} = [font=\small, black]
\tikzstyle{edge} = [draw,-]
\tikzstyle{selected edge} = [draw,line width=3pt,-,gray!70]
\tikzstyle{matched edge} = [draw,line width=3pt,-]
\tikzstyle{dashed edge} = [draw,dashed]
\tikzstyle{selected vertex} = [vertex, fill=black!70, text=white]

\tikzstyle{remove selected vertex} = [vertex, fill=white!100]
\tikzstyle{blue colored edge} = [draw,line width=3pt,-,blue!80]
\tikzstyle{blue colored curved edge left} = [draw,line width=3pt,-,blue!80,bend left]
\tikzstyle{blue colored curved edge right} = [draw,line width=3pt,-,blue!80,bend right]
\tikzstyle{green colored edge} = [draw,line width=3pt,-,green!80]
\tikzstyle{remove selected edge} = [draw,line width=4pt,-,white!100]
\tikzstyle{remove matched edge} = [draw,line width=4pt,-,white!80,bend right]

\SetArgSty{textup}
\DontPrintSemicolon

\newlength{\rad}
\title{\textbf{Stabilizing Weighted Graphs}}

\author{Zhuan Khye Koh \and Laura Sanit\`{a}\\}

\date{\normalsize{Combinatorics and Optimization, University of Waterloo\\
Waterloo, ON N2L 3G1, Canada\\
\tt{$\{$zkkoh,lsanita$\}$@uwaterloo.ca}}}

\begin{document}

\maketitle

\begin{abstract}
An edge-weighted graph $G=(V,E)$ is called \emph{stable} if  the value of a maximum-weight matching equals the value of a maximum-weight \emph{fractional} matching. Stable graphs play an important role in some interesting game theory problems, such as \emph{network bargaining} games and \emph{cooperative matching} games, because they characterize instances which admit stable outcomes. Motivated by this, in the last few years many researchers have investigated the algorithmic problem of turning a given graph into a stable one, via edge- and vertex-removal operations. However, all the algorithmic results developed in the literature so far only hold for \emph{unweighted} instances, i.e., assuming unit weights on the edges of $G$.

We give the first polynomial-time algorithm to find a minimum cardinality subset of vertices whose removal from $G$ yields a stable graph, for any weighted graph $G$. The algorithm is combinatorial and exploits new structural properties of basic fractional matchings, which are of independent interest. In particular, one of the main ingredients of our result is the development of a polynomial-time algorithm to compute
a basic maximum-weight fractional matching with minimum number of \emph{odd cycles} in its support. This generalizes
a fundamental and classical result on unweighted matchings given by Balas more than 30 years ago, which we expect to prove useful beyond this particular application.

In contrast, we show that the problem of finding a minimum cardinality subset of edges whose removal from a weighted graph $G$ yields a stable graph, does not admit any constant-factor approximation algorithm, unless $P=NP$. In this setting, we develop an $O(\Delta)$-approximation algorithm for the problem, where $\Delta$ is the maximum degree of a node in $G$.
\end{abstract}

\section{Introduction}

Several interesting game theory problems are defined on networks, where the vertices 
represent players and the edges model the way players 
can interact with each other. In many such games, the structure of the underlying graph that describes the interactions among players is essential in determining the existence of stable outcomes for the corresponding games, i.e., outcomes where players have no incentive to deviate.
Popular examples are \emph{cooperative matching} games, introduced by Shapley and Shubik~\cite{journals/ijgt/Shapley71}, and 
\emph{network bargaining} games, defined by Kleinberg and Tardos~\cite{conf/stoc/KleinbergT08}, both extensively studied in the game theory community. 
Instances of such games are described by a graph $G=(V,E)$ with edge weights $w \in \mathbb R^E_{\geq 0}$, where $V$ represents a set of players, and the value of a \emph{maximum-weight matching}, denoted as $\nu(G)$, is the total value that the players could get by interacting with each other. 

An important role in such games is played by so-called \emph{stable} graphs.  An edge-weighted graph $G=(V,E)$ is called stable if 
the value $\nu(G)$ of a maximum-weight matching equals the value of a maximum-weight \emph{fractional} matching, denoted as $\nu_f(G)$. Formally, $\nu_f(G)$ is given by the optimal value of the standard linear programming relaxation of the matching problem, defined as
\begin{equation}\tag{P}
  \nu_f(G) := \max\set{w^{\top}x:x(\delta(v))\leq 1 \; \forall v\in V, x\geq 0}
\end{equation}
Here $x$ is a vector in $\mathbb R^E$, $\delta(v)$ denotes the set of edges incident to the node $v$, and for a set $F\subseteq E$, 
$x(F)=\sum_{e\in F} x_e$. Feasible solutions of the above LP are called \emph{fractional matchings}.

The relation that interplays between stable graphs and network games is as follows. In cooperative matching games \cite{journals/ijgt/Shapley71}, the goal is to find an allocation of the value $\nu(G)$ among the vertices, given as a vector $y \in \mathbb R^V_{\geq 0}$, such that no subset $S \subseteq V$ has an incentive to form a \emph{coalition} to deviate. This condition is formally defined by the constraints $\sum_{v \in S} y_v \geq \nu(G[S]), \forall S \subseteq V$, where $G[S]$ denotes the subgraph induced by $S$, and an allocation $y$ that satisfies the above set of constraints is called \emph{stable}. Deng et al.~\cite{journals/mor/DengIN99} proved that a stable allocation exists if and only if the graph describing the game is a stable graph. This is an easy consequence from LP duality. If $y$ is a stable allocation, then $y$ is a feasible dual solution of value $\nu(G)$, showing that $\nu_f(G)=\nu(G)$. Conversely, if $\nu_f(G)=\nu(G)$, then an optimal dual solution yields a stable allocation of $\nu(G)$.

In network bargaining games \cite{conf/stoc/KleinbergT08}, each edge $e$ represents a deal of value $w_e$. A player can enter in a deal with at most one neighbor, and when a deal is made, the players have to agree on how to split the value of the deal between them. An outcome of the game is given by a pair $(M, y)$, where $M$ is a matching of $G$ and stands for the set of deals made by the players, and $y \in \mathbb R^V_{\geq 0}$ is an allocation vector representing how the deal values have been split. Kleinberg and Tardos have defined a notion of \emph{stable} outcome for such games, as well as a notion of \emph{balanced} outcome, that are outcomes where players have no incentive to deviate, and in addition the deal values are ``fairly'' split among players. They proved that a balanced outcome exists if and only if a stable outcome  exists, and this happens if and only if the graph $G$ describing the game is stable.

Motivated by the above connection, in the last few years many researchers have investigated the algorithmic problem of turning a given graph into a stable one, by performing a minimum number of modifications on the input graph \cite{journals/mp/BockCKPS15,conf/ipco/AhmadianHS16,journals/tcs/ItoKKKO17,arxiv/ChandrasekaranG16,journals/mst/KonemannL015,journals/tcs/BiroBGKP,journals/ijgt/BiroKP12}. Two natural operations which have a nice network game interpretation, are vertex-deletion and edge-deletion. They correspond to \emph{blocking players} and \emph{blocking deals}, respectively, in order to achieve stability in the corresponding games. Formally,  a subset of vertices $S \subseteq V$ is called a \emph{vertex-stabilizer} if the graph $G\setminus S := G[V\setminus S]$ is stable.
Similarly, a subset of edges $F \subseteq E$ is called an \emph{edge-stabilizer} if the graph $G\setminus F := (V, E\setminus F)$ is stable.
The corresponding optimization problems, which are the focus of this paper, are: 

\smallskip
\noindent
{\bf Minimum Vertex-stabilizer:} \emph{Given an edge-weighted graph $G=(V,E)$, find a minimum-cardinality vertex-stabilizer.}\\
{\bf Minimum Edge-stabilizer:} \emph{Given an edge-weighted graph $G=(V,E)$, find a minimum-cardinality edge-stabilizer.}

\smallskip
The above problems have been studied quite intensively in the last few years on unweighted graphs.
In particular, Bock et al.~\cite{journals/mp/BockCKPS15} have showed that finding a minimum-cardinality edge-stabilizer is hard to approximate within a factor of $(2-\varepsilon)$, assuming Unique Game Conjecture (UGC) \cite{conf/stoc/Khot02}. On the positive side, 
they have given an approximation algorithm for the edge-stabilizer problem, whose approximation factor depends on the sparsity of the input graph $G$. In other work, Ahmadian et al.~\cite{conf/ipco/AhmadianHS16} and Ito et al.~\cite{journals/tcs/ItoKKKO17} have shown independently that finding a minimum-cardinality vertex-stabilizer is a polynomial-time solvable problem. These (exact and approximate) algorithmic results, developed for unweighted instances, do not easily generalize when dealing with arbitrary edge-weights, since they heavily rely on the structure of maximum matchings in unweighted graphs. In fact, unweighted instances of the above problems exhibit a very nice property, as shown in \cite{journals/mp/BockCKPS15,conf/ipco/AhmadianHS16}: the removal of any inclusion-wise minimal edge-stabilizer (resp. vertex-stabilizer) from a graph $G$ \emph{does not decrease} the cardinality of a maximum matching in the resulting graph. This property ensures that there is at least one maximum-cardinality matching that survives in the modified graph, and this insight can be successfully exploited when designing (exact and approximate) algorithms. Unfortunately, it is not difficult to realize that this crucial property does not hold anymore when dealing with edge-weighted graphs (see Appendix I), and in fact, the development of algorithmic results for weighted graphs requires substantial new ideas.

\paragraph{Our results and techniques.} Our main results are as follows.

\smallskip
\noindent
{\bf Vertex-stabilizers}. We give the first polynomial-time algorithm to find a minimum-cardinality vertex-stabilizer $S$, in any weighted graph $G$. Our algorithm also ensures that $\nu(G\setminus S) \geq \frac{2}{3}\nu(G)$, i.e., the value of a maximum-weight matching
is preserved up to a factor of $\frac{2}{3}$, and we show that this factor is tight in general. Specifically, as previously mentioned, a minimum-cardinality vertex-stabilizer for a weighted graph might decrease the value of a maximum-weight matching in the resulting graph. From a network bargaining perspective, this means we are decreasing the total value which the players are able to get, which is of course undesirable. However, we can show this is inevitable, since deciding whether there exists \emph{any} vertex-stabilizer $S$ that preserves the value of a maximum-weight matching (i.e., such that $\nu(G\setminus S)=\nu(G)$) is an NP-hard problem. Furthermore, we give an example of a graph $G$ where \emph{any} vertex-stabilizer $S$ decreases the value of a maximum-weight matching by a factor of essentially $\frac{1}{3}$, i.e. $\nu(G\setminus S) \leq \pr{\frac{2}{3} + \varepsilon} \nu(G)$ (for an arbitrarily small $\varepsilon >0$). This shows that the bounds of our algorithm are essentially best possible: the algorithm finds a vertex-stabilizer $S$ whose cardinality is the \emph{smallest} possible, and preserves the value of a maximum-weight matching up to a factor of $\frac{2}{3}$, that is the tightest factor that holds for all instances.

The above result is based on two main ingredients. The first one is giving a lower bound on the cardinality of a minimum vertex-stabilizer, which generalizes the lower bound used in the unweighted setting, and is based on the structure of optimal basic solutions of (P). In particular, it was shown in \cite{conf/ipco/AhmadianHS16} that a lower bound on the cardinality of a vertex-stabilizer for unweighted graphs is given by the minimum number of \emph{odd-cycles} in the support of an optimal basic solution to (P). We show that this lower bound holds also for weighted graphs, though 
this generalization is not obvious (in fact, as we will show later, the same generalization does \emph{not} hold for edge-stabilizers). Consequently, our proof is much more involved, and requires different ideas. 
The second main ingredient is giving a polynomial-time algorithm for computing an optimal basic solution to (P) with the smallest number of odd-cycles in its support, which is of independent interest, as highlighted in the next paragraph.

\smallskip
\noindent
{\bf Computing maximum fractional matchings with minimum cycle support.} 
The fractional matching polytope given by (P) has been
extensively studied in the literature, and characterizing 
instances for which a maximum fractional matching equals an integral one is a natural graph theory question
(see \cite{journals/mp/BockCKPS15,conf/ipco/AhmadianHS16}). It is well-known that basic solutions of (P)
are half-integral, and the support of a basic solution is the disjoint union of a matching (given by 1-valued entries) and a set of odd-cycles (given by half-valued entries). 
Balas \cite{journal/nms/Balas81} gave a nice polynomial-time algorithm to compute a basic maximum fractional matching in an unweighted graph, with minimum number of odd-cycles in its support. This is a classical  result on matching theory, which has been known for more than 30 years. In this paper, we generalize this result to arbitrary weighted instances, exploiting nice structural properties of basic fractional matchings. Our algorithm is based on combinatorial techniques, and we expect that this result will prove useful beyond this particular application. 

\smallskip
\noindent
{\bf Edge-stabilizers}. When dealing with edge-removal operations, the stabilizer problem becomes harder, already in the unweighted setting.
It is shown in \cite{journals/mp/BockCKPS15} that finding a minimum edge-stabilizer is as hard as vertex cover, and whether the problem admits a constant factor approximation algorithm is an interesting open question. We here show that the answer to this question is negative for weighted graphs, since we prove that the minimum edge-stabilizer problem for a weighted graph $G$ does not admit any constant-factor approximation algorithm, unless $P=NP$. From an approximation point of view, we show that the algorithm we developed for the vertex-stabilizer problem translates into a $O(\Delta)$-approximation algorithm for the edge-stabilizer problem, where $\Delta$ is the maximum degree of a node in $G$.

Once again, the analysis relies on proving a lower bound on the cardinality of a minimum edge-stabilizer. It was shown in \cite{journals/mp/BockCKPS15} that a lower-bound on the cardinality of a minimum edge-stabilizer for unweighted graphs is again given by the minimum number of odd-cycles in the support of an optimal solution to (P) (called $\gamma(G)$). Interestingly, we show that, differently from the vertex-stabilizer setting, here this lower bound does not generalize, and $\gamma(G)$ \emph{is not} a lower bound on the cardinality of an edge-stabilizer for arbitrary weighted graphs. However, we are able to show that $\ceil{\gamma(G)/2}$ \emph{is} a lower bound on the cardinality of a minimum edge-stabilizer, and this is enough for our approximation purposes.

\smallskip
\noindent
{\bf Additional results}. Lastly, we also generalize a result given in \cite{conf/ipco/AhmadianHS16} on finding a 
 minimum vertex-stabilizer which avoids a fixed maximum matching $M$, on unweighted graphs. We prove that if $M$ is a maximum-weight matching of a weighted graph $G$, then finding a minimum vertex-stabilizer that is element-disjoint from $M$ is a polynomial-time solvable problem. Otherwise, if $M$ is not a maximum-weight matching, the problem is at least as hard as vertex cover. We supplement this result with a 2-approximation algorithm for this case, that is best possible assuming UGC.

\paragraph{Related work.}
Bir\'o et al.~\cite{journals/tcs/BiroBGKP} were the first to consider the edge-stabilizer problem in weighted graphs,
and they showed NP-hardness for this case. Stabilizing a graph
via different operations on the input graph (other than removing edges/vertices) has also been
studied. In particular, Ito et al.~\cite{journals/tcs/ItoKKKO17} have given polynomial-time algorithms to stabilize an unweighted graph by adding edges and by adding vertices. Chandrasekaran et al.~\cite{arxiv/ChandrasekaranG16} have recently studied 
the problem of stabilizing unweighted graphs by fractionally increasing edge weights.
Ahmadian et al.~\cite{conf/ipco/AhmadianHS16} have also studied the vertex-stabilizer problem on unweighted
graphs, but in the more-general setting where there are (non-uniform) costs for removing vertices, and gave approximation algorithms for this case. 

Bir\'o et al.~\cite{journals/ijgt/BiroKP12} and K\"onemann et al.~\cite{journals/mst/KonemannL015} studied a variant of the problem where the goal is to compute a minimum-cardinality set of \emph{blocking pairs}, that are edges whose removal from the graph yield the existence of a fractional vertex cover of size at most $\nu(G)$ (but note that the resulting graph might not be stable). Mishra et al.~\cite{journals/algorithmica/MishraRSSS11} studied the problem of converting a graph into a \emph{K\"onig-Egerv\'ary graph}, via vertex-deletion and edge-deletion operations. A K\"onig-Egerv\'ary graph is a graph where the size of a maximum matching equals the size of an (integral) minimum vertex cover. They gave an $O(\log n \log \log n)$-approximation algorithm for the vertex-removal setting in unweighted graphs, and showed constant-factor hardness of approximation (assuming UGC) for 
both the minimum vertex-removal and edge-removal problem.

\paragraph{Paper Organization.} In Section \ref{sec:preliminaries}, we give some preliminaries and discuss notation. 
In Section \ref{sec:frac_matching}, we give a polynomial-time algorithm to compute an optimal basic solution to (P) with minimum number
of odd cycles in its support. This algorithm will be crucially used in Section \ref{sec:vertex_stabilizer}, where we give our
results on vertex-stabilizers. Section \ref{sec:edge_stabilizer} reports our
results on edge-stabilizers. Finally, our additional results can be found in Section \ref{sec:additional}.

\section{Preliminaries and notation}
\label{sec:preliminaries}

A key concept that we will use is LP duality. The dual of (P) is given by
\begin{equation}\tag{D}
  \tau_f(G) := \min\set{\1^{\top}y:y_u+y_v\geq w_{uv} \; \forall uv\in E, y\geq 0}.
\end{equation}
As feasible solutions to (P) are called fractional matchings, we call feasible solutions to (D) \emph{fractional $w$-vertex covers}. In fact,
(D) is the standard LP-relaxation of the problem of finding a minimum $w$-vertex cover, obtained by 
adding integrality constraints on (D). We also call basic feasible solutions to (P) as \emph{basic fractional matchings}. An application of duality theory yields the following relationship
$\nu(G) \leq \nu_f(G) = \tau_f(G)$. Recall that a graph $G$ is \emph{stable} if $\nu(G) = \nu_f(G) = \tau_f(G)$.

For a vector $x\in \R^{E}$ and any subset $F\subseteq E$, we denote $x_{-F}\in \R^{E -F}$ as the subvector obtained by dropping the entries corresponding to $F$. For any multisubset $F\subseteq E$, we define $x(F):=\sum_{e\in F}x_e$. Note that an element may be accounted for multiple times in the sum if it appears more than once in $F$. We denote $\supp(x):=\set{e\in E:x_e\neq 0}$ as the support of $x$. For any positive integer $k$, $[k]$ represents the set $\set{1,2,\dots,k}$.

Given an undirected graph $G$, we denote by $n$ the number of vertices and by $m$ the number of edges. For edge weights $w\in\R^m_+$ and a matching $M$ in $G$, a path/walk is called \emph{$M$-alternating} if its edges alternately belong to $M$ and $E\setminus M$. Recall that a walk is a path that is non-simple. We say that an $M$-alternating path/walk is \emph{valid} if it starts with an $M$-exposed vertex or an edge in $M$, and ends with an $M$-exposed vertex or an edge in $M$ (see Figure \ref{fig:valid} for an example). A valid $M$-alternating path/walk $P$ is called \emph{$M$-augmenting} if $w(P\setminus M)>w(P\cap M)$. A cycle is also called $M$-alternating if its edges alternately belong to $M$ and $E\setminus M$. Note that an $M$-alternating cycle has even length. An $M$-alternating cycle $C$ is said to be $M$-augmenting if $w(C\setminus M)>w(C\cap M)$.

\begin{figure}[H]
\centering
\def\dist{1.2}
\begin{tikzpicture}[node distance=\dist cm, inner sep=2.5pt, minimum size=2.5pt, auto]
% Nodes
\node[vertex] (a1) {};
\node[vertex] (a2) [right of=a1] {};
\node[vertex] (a3) [right of=a2] {};
\node[vertex] (a4) [right of=a3] {};
\node[vertex] (a5) [right of=a4] {};
\node[vertex] (a6) [right of=a5] {};	

% Edges
\path[matched edge] (a1) -- (a2);
\path[edge] (a2) -- (a3);
\path[matched edge] (a3) -- (a4);
\path[edge] (a4) -- (a5);
\path[matched edge] (a5) -- (a6);

% Braces
\draw[decoration={brace,raise=12pt},decorate] (0,0) -- node[weight,above=18pt] {$P$} (5*\dist,0);
\draw[decoration={brace,raise=12pt},decorate] (4*\dist,0) -- node[weight,below=18pt] {$P'$} (0,0);
\end{tikzpicture}

\caption{An example of a valid and invalid alternating path. Here, $P$ is valid while $P'$ is invalid.}
\label{fig:valid}
\end{figure}
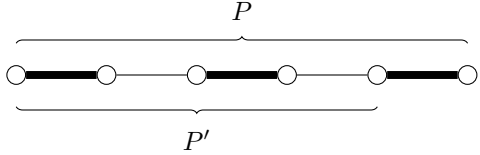

\begin{definition}
An odd cycle $C=(e_1,e_2,\dots,e_{2k+1})$ is called an $M$-\emph{blossom} if $e_i\in M$ for all even $i$ and $e_i\notin M$ for all odd $i$. The vertex $v:=e_1\cap e_{2k+1}$ is called the \emph{base} of the blossom. The blossom is \emph{augmenting} if $v$ is $M$-exposed and $w(C\setminus M)>w(C\cap M)$.
\end{definition}

\begin{definition}\label{def:flower}
An \emph{$M$-flower} $C\cup P$ consists of an $M$-blossom $C$ with base $v_1$ and a valid $M$-alternating path $P=(v_1,v_2,\dots,v_k)$ where $v_1v_2\in M$. The vertex $v_k$ is called the \emph{root} of the flower. The flower is \emph{augmenting} if 
\[w(C\setminus M)+2w(P\setminus M)>w(C\cap M)+2w(P\cap M).\]
\end{definition}

Given an $M$-augmenting flower $C\cup P$, if we replace the vector which places 1 on the edges of $M\cap (C\cup P)$, with the vector that places $\frac{1}{2}$ on the edges of C and 1 on the edges of $P\setminus M$, then the change in weight is exactly $\frac{1}{2}$ times LHS $-$ RHS of the above inequality. So the inequality means that this operation increases the weight. 

\begin{definition}\label{def:bicycle}
An \emph{$M$-bi-cycle} $C\cup P\cup D$ consists of two $M$-blossoms $C,D$ with bases $v_1,v_k$ respectively and an odd $M$-alternating path $P=(v_1,v_2,\dots,v_k)$ where $v_1v_2,v_{k-1}v_k\in M$. The bi-cycle is \emph{augmenting} if 
\[w(C\setminus M)+2w(P\setminus M)+w(D\setminus M)>w(C\cap M)+2w(P\cap M)+w(D\cap M).\]
\end{definition}

Note that the structures defined in Definition \ref{def:flower} and \ref{def:bicycle} might not be simple. For example, in a flower $C\cup P$, the path $P$ might intersect the blossom $C$ more than once. Figure \ref{fig:simple} illustrates some simple examples of these structures. Notice that a blossom is always simple. 

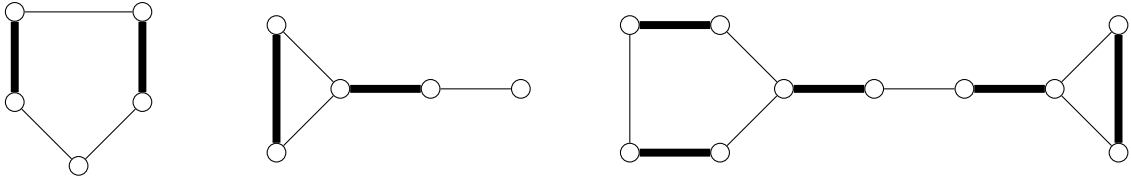
\begin{figure}[ht]
\centering
\def\dist{1.2}
\begin{minipage}{0.2\textwidth}
\centering
\begin{tikzpicture}[node distance=\dist cm, inner sep=2.5pt, minimum size=2.5pt, auto]
% Nodes
\node[vertex] (a1) {};
\node[vertex] (a2) [above left of=a1] {};
\node[vertex] (a3) [above right of=a1] {};
\node[vertex] (a4) [above of=a2] {};
\node[vertex] (a5) [above of=a3] {};	

% Edges
\path[edge] (a1) -- (a2);
\path[edge] (a1) -- (a3);
\path[matched edge] (a2) -- (a4);
\path[matched edge] (a3) -- (a5);
\path[edge] (a4) -- (a5);
\end{tikzpicture}
\end{minipage}
\begin{minipage}{0.3\textwidth}
\centering
\begin{tikzpicture}[node distance=\dist cm, inner sep=2.5pt, minimum size=2.5pt, auto]
% Nodes
\node[vertex] (a1) {};
\node[vertex] (a2) [above left of=a1] {};
\node[vertex] (a3) [below left of=a1] {};
\node[vertex] (a4) [right of=a1] {};
\node[vertex] (a5) [right of=a4] {};	

% Edges
\path[edge] (a1) -- (a2);
\path[edge] (a1) -- (a3);
\path[matched edge] (a2) -- (a3);
\path[matched edge] (a1) -- (a4);
\path[edge] (a4) -- (a5);
\end{tikzpicture}
\end{minipage}
\begin{minipage}{0.45\textwidth}
\centering
\begin{tikzpicture}[node distance=\dist cm, inner sep=2.5pt, minimum size=2.5pt, auto]
% Nodes
\node[vertex] (a1) {};
\node[vertex] (a2) [above left of=a1] {};
\node[vertex] (a3) [below left of=a1] {};
\node[vertex] (a4) [left of=a2] {};
\node[vertex] (a5) [left of=a3] {};
\node[vertex] (a6) [right of=a1] {};
\node[vertex] (a7) [right of=a6] {};
\node[vertex] (a8) [right of=a7] {};
\node[vertex] (a9) [above right of=a8] {};
\node[vertex] (a10) [below right of=a8] {};	

% Edges
\path[edge] (a1) -- (a2);
\path[edge] (a1) -- (a3);
\path[matched edge] (a2) -- (a4);
\path[matched edge] (a3) -- (a5);
\path[edge] (a4) -- (a5);
\path[matched edge] (a1) -- (a6);
\path[edge] (a6) -- (a7);
\path[matched edge] (a7) -- (a8);
\path[edge] (a8) -- (a9);
\path[edge] (a8) -- (a10);
\path[matched edge] (a9) -- (a10);
\end{tikzpicture}
\end{minipage}

\caption{Simple examples of a blossom, a flower and a bi-cycle.}
\label{fig:simple}
\end{figure}
\
The significance of the structures defined above is given by the following theorem:

\begin{theorem}[\cite{conf/stoc/KleinbergT08}]\label{thm:stable}
If a graph is stable, then it does not have an $M$-augmenting flower or bi-cycle for every maximum-weight matching $M$. Otherwise, it has an $M$-augmenting flower or bi-cycle for every maximum-weight matching $M$.
\end{theorem}

We will need the following classical result on the structure of basic fractional matchings:

\begin{theorem}[\cite{conf/smp/Balinski70}]\label{thm:basic}
A fractional matching $x$ in $G=(V,E)$ is basic if and only if $x_e=\set{0,\frac{1}{2},1}$ for all $e\in E$ and the edges $e$ having $x_e=\frac{1}{2}$ induce vertex-disjoint odd cycles in $G$.
\end{theorem}

Let $\hat{x}$ be a basic fractional matching in $G$. We partition the support of $\hat{x}$ into two parts. Define
\[\mathscr{C}(\hat{x}) := \set{C_1,\dots,C_q} \quad \text{and} \quad M(\hat{x}) := \set{e\in E:\hat{x}_e=1}\]
as the set of odd cycles such that $\hat{x}_e=\frac{1}{2}$ for all $e\in E(C_i)$ and the set of matched edges in $\hat{x}$ respectively. For ease of notation, we use $V(\mathscr{C}(\hat{x})) = \cup_{C\in \mathscr{C}(\hat{x})}V(C)$ and $E(\mathscr{C}(\hat{x})) = \cup_{C\in \mathscr{C}(\hat{x})}E(C)$ to denote the vertex set and edge set of $\mathscr{C}(\hat{x})$ respectively. We define two operations on the entries of $\hat{x}$ associated with certain edge sets of $G$:

\begin{definition}
By \emph{complementing} on $E'\subseteq E$, we mean replacing $\hat{x}_e$ by $\bar{x}_e=1-\hat{x}_e$ for all $e\in E'$.
\end{definition}

\begin{definition}
By \emph{alternate rounding} on $C\in \mathscr{C}(\hat{x})$ at $v$ where $C=\set{e_1,\dots,e_{2k+1}}$ and $v=e_1\cap e_{2k+1}$, we mean replacing $\hat{x}_e$ by $\bar{x}_e = 0$ for all $e\in \set{e_1,e_3,\dots,e_{2k+1}}$ and $\bar{x}_e=1$ for all $e\in \set{e_2,e_4,\dots,e_{2k}}$. When $v$ is clear from the context, we just say alternate rounding on $C$.
\end{definition}

Let $\mathcal{X}$ be the set of basic maximum-weight fractional matchings in $G$. Define
\[\gamma(G) := \min_{\hat{x}\in \mathcal{X}}\size{\mathscr{C}(\hat{x})}.\]
Note that $G$ is stable if and only if $\gamma(G)=0$.

We will use the following terminology given in \cite{book/CookCPS98} for the description of Edmonds' maximum matching algorithm. Given a graph $G$ and a matching $M$, let $T$ be an $M$-alternating tree rooted at a vertex $r$. We denote by $A(T)$ and $B(T)$ the sets of nodes in $T$ at odd and even distance respectively from $r$. We call $T$ \emph{frustrated} if every edge of $G$ having one end in $B(T)$ has the other end in $A(T)$.

Finally, the following theorem gives a sufficient condition for a graph to be Hamiltonian.

\begin{theorem}[Ore's Theorem \cite{journals/amm/Ore60}]
Let $G$ be a finite and simple graph with $n\geq 3$ vertices. If $\deg(u)+\deg(v)\geq n$ for every pair of distinct non-adjacent vertices $u$ and $v$, then $G$ is Hamiltonian.
\end{theorem}

\section{Maximum fractional matching with minimum support}
\label{sec:frac_matching}
 
In this section, we give a polynomial-time algorithm to compute a basic maximum-weight fractional matching $\hat{x}$ for a weighted graph $G$ with minimum number of odd cycles in its support, i.e., satisfying $\size{\mathscr{C}(\hat{x})} = \gamma(G)$.
This algorithm will be used as a subroutine by our vertex-stabilizer algorithm, which we will develop in Section \ref{sec:vertex_stabilizer}.

Our first step is to characterize basic maximum-weight fractional matchings which have more than $\gamma(G)$ odd cycles. Balas~\cite{journal/nms/Balas81} considered this problem on unweighted graphs, and gave the following characterization:

\begin{theorem}[\cite{journal/nms/Balas81}]
Let $\hat{x}$ be a basic maximum fractional matching in an unweighted graph $G$. If $\size{\mathscr{C}(\hat{x})}>\gamma(G)$, then there exists an $M(\hat{x})$-alternating path which connects two odd cycles $C_i,C_j \in \mathscr{C}(\hat{x})$. Furthermore, alternate rounding on the odd cycles and complementing on the path produces a basic maximum fractional matching $\bar{x}$ such that $\mathscr{C}(\bar{x})\subset\mathscr{C}(\hat{x})$.
\end{theorem}

We generalize this to weighted graphs. Before stating the theorem, we need to introduce the concept of \emph{connector} (see Figure \ref{fig:connector} for some examples):

\begin{definition}
Let $C$ be a cycle and $S_0,S_1,\dots,S_k$ be a partition of $V(C)$ such that $\size{S_0}$ is even and $k\geq 2$, where $S_0$ is allowed to be empty. 
Let $M$ be a perfect matching on the vertex set $S_0$. We call the graph $C\cup M$ a \emph{connector}. Each $S_i$ is called a \emph{terminal set} for $i\geq 1$. An edge $e\in M$ is called a \emph{chord} if $e\notin E(C)$. 
\end{definition}

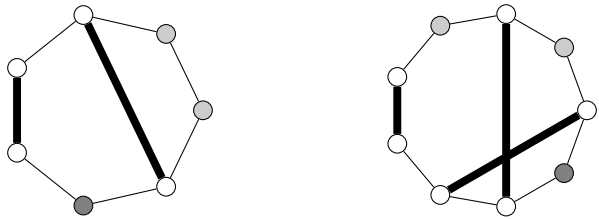
\begin{figure}[ht]
\centering
\begin{minipage}{0.3\textwidth}
\centering
\begin{tikzpicture}[inner sep=2.5pt, minimum size=2.5pt, auto]
	\def \n {7}
	\def \radius {1.3cm}
	% Nodes
	\node[vertex,fill=gray!40] (u1) at ({360/\n*1}:\radius) {};
	\node[vertex] (u2) at ({360/\n*2}:\radius) {};
	\node[vertex] (u3) at ({360/\n*3}:\radius) {};
	\node[vertex] (u4) at ({360/\n*4}:\radius) {};
	\node[vertex,fill=gray] (u5) at ({360/\n*5}:\radius) {};
	\node[vertex] (u6) at ({360/\n*6}:\radius) {};
	\node[vertex,fill=gray!40] (u7) at ({360/\n*7}:\radius) {};

	% Edges
	\path[edge] (u1) -- (u2);
	\path[edge] (u2) -- (u3);
	\path[matched edge] (u3) -- (u4);
	\path[edge] (u4) -- (u5);
	\path[edge] (u5) -- (u6);
	\path[edge] (u6) -- (u7);
	\path[edge] (u7) -- (u1);
	\path[matched edge] (u2) -- (u6);
\end{tikzpicture}
\end{minipage}
\begin{minipage}{0.3\textwidth}
\centering
\begin{tikzpicture}[inner sep=2.5pt, minimum size=2.5pt, auto]
	\def \n {9}
	\def \radius {1.3cm}
	% Nodes
	\node[vertex,fill=gray!40] (v1) at ({360/\n*1}:\radius) {};
	\node[vertex] (v2) at ({360/\n*2}:\radius) {};
	\node[vertex,fill=gray!40] (v3) at ({360/\n*3}:\radius) {};
	\node[vertex] (v4) at ({360/\n*4}:\radius) {};
	\node[vertex] (v5) at ({360/\n*5}:\radius) {};
	\node[vertex] (v6) at ({360/\n*6}:\radius) {};
	\node[vertex] (v7) at ({360/\n*7}:\radius) {};
	\node[vertex,fill=gray] (v8) at ({360/\n*8}:\radius) {};
	\node[vertex] (v9) at ({360/\n*9}:\radius) {};

	% Edges
	\path[edge] (v1) -- (v2);
	\path[edge] (v2) -- (v3);
	\path[edge] (v3) -- (v4);
	\path[matched edge] (v4) -- (v5);
	\path[edge] (v5) -- (v6);
	\path[edge] (v6) -- (v7);
	\path[edge] (v7) -- (v8);
	\path[edge] (v8) -- (v9);
	\path[edge] (v9) -- (v1);
	\path[matched edge] (v2) -- (v7);
	\path[matched edge] (v6) -- (v9);
\end{tikzpicture}
\end{minipage}

\caption{Two examples of connectors. Bold edges indicate $M$. Vertices of the same color belong to the same terminal set. White vertices are the ones in $S_0$.}
\label{fig:connector}
\end{figure}

Connectors are useful because of the following property:

\begin{lemma}
\label{lem:connector}
Let $C\cup M$ be a connector. For every terminal set $S_i$, there exists an $M$-augmenting path in the connector from a vertex $v \in S_i$ to a vertex $u \in S_j$, for some $j\neq i$.
\end{lemma}

\begin{proof}
For every $e\in M\cap E(C)$, contract $e$ and smooth away the vertex formed after the contraction (smoothing is the reverse operation of subdivision). The only edges in $M$ that survive this process are the chords. Fix $i$ and identify all the vertices in $S_i$ into a single vertex $v_i$. Denote the resulting (multi)graph as $G=(V,E)$. Observe that there exists an $M$-augmenting path from $S_i$ to $S_j$ in $C\cup M$ if and only if there exists an $M$-augmenting path from $v_i$ to $S_j$ in $G$ where $i\neq j$. Hence, we will work on the reduced graph $G$.

Apply Edmonds' maximum matching algorithm on $G$ initialized with the matching $M\cap E$, and construct an $M$-alternating tree starting with the exposed vertex $v_i$. There are two possibilities: either we find an augmenting path from $v_i$ to $S_j$ for some $j\neq i$ or a frustrated tree rooted at $v_i$. For the purpose of contradiction, suppose we get a frustrated tree $T$ rooted at $v_i$. Let $\widetilde{T}=T\cup D$, where $D=\set{uv\notin E(T):u\in A(T), v\in B(T)}$. Note that we do not have edges connecting two nodes in $B(T)$, otherwise $T$ is not a frustrating tree. 

We claim that each pseudonode in $T$ is incident to at least two unmatched edges in $\widetilde{T}$. Let $v$ be a pseudonode in $T$, and $S(v)$ be the subset of vertices in $G$ that are contained in $v$ (after expanding pseudonodes). Note that $S(v)\subset V$ because there are at least two exposed vertices in $G$. Let $\delta_{G\setminus M}(\cdot)$ denote the cut function on $G\setminus M$. Since $G\setminus M$ is 2-edge-connected, we have $\size{\delta_{G\setminus M}(S(v))}\geq 2$. These edges are present in $\widetilde{T}$ because otherwise we can extend the alternating tree $T$. It follows that $v$ is incident to at least two unmatched edges in $\widetilde{T}$.

Let $uv$ be a matched edge in $T$ where $u\in A(T)$ and $v\in B(T)$. We claim that $\deg_{\widetilde{T}}(u)\leq \deg_{\widetilde{T}}(v)$. Note that $\deg_{\widetilde{T}}(u)$ is either 2 or 3. This is because $u$ is not a pseudonode, and $\deg_G(w)=3$ for every $M$-covered vertex $w$ in $G$. If $v$ is not a pseudonode, then $\deg_{\widetilde{T}}(v)=3$ as all edges in $\delta_G(v)$ are accounted for in $\widetilde{T}$. Otherwise, if $v$ is a pseudonode, then by the previous claim $v$ is incident to at least two unmatched edges in $\widetilde{T}$. So $\deg_{\widetilde{T}}(v)\geq 3$.

Now, observe that $\widetilde{T}$ is a bipartite graph as the node set can be partitioned into $A(T)$ and $B(T)$ where $\size{B(T)}=\size{A(T)}+1$. For every $v\in A(T)$, let $M(v)$ be its matched neighbour in $B(T)$. The extra node in $B(T)$ is the root of $T$, which has degree at least one in $\widetilde{T}$. Summing up the node degrees in $A(T)$, we obtain
\[\sum_{v\in A(T)}\deg_{\widetilde{T}}(v) \leq \sum_{v\in A(T)}\deg_{\widetilde{T}}(M(v)) < \sum_{v\in A(T)}\deg_{\widetilde{T}}(M(v)) + 1 \leq \sum_{v\in B(T)}\deg_{\widetilde{T}}(v)\]
which is a contradiction.  
\end{proof}

Let $y$ be a minimum fractional $w$-vertex cover in $G$. We say that an edge $uv$ is \emph{tight} if $y_u+y_v=w_{uv}$. Similarly, we say that a path is tight if all of its edges are tight.

\begin{theorem}\label{thm:cycles}
Let $\hat{x}$ be a basic maximum-weight fractional matching and $y$ be a minimum fractional $w$-vertex cover in $G$. If $\size{\mathscr{C}(\hat{x})}> \gamma(G)$, then there exists
\begin{enumerate}[noitemsep,topsep=0pt]
  \item[(i)] a vertex $v\in V(C_i)$ for some odd cycle $C_i\in \mathscr{C}(\hat{x})$ such that $y_v=0$; or
	\item[(ii)] a tight $M(\hat{x})$-alternating path $P$ which connects two odd cycles $C_i,C_j\in\mathscr{C}(\hat{x})$; or
	\item[(iii)] a tight and valid $M(\hat{x})$-alternating path $P$ which connects an odd cycle $C_i\in\mathscr{C}(\hat{x})$ and a vertex $v\notin V(\mathscr{C}(\hat{x}))$ such that $y_v=0$.
\end{enumerate}
Furthermore, alternate rounding on the odd cycles and complementing on the path produces a basic maximum-weight fractional matching $\bar{x}$ such that $\mathscr{C}(\bar{x})\subset \mathscr{C}(\hat{x})$.
\end{theorem}

\begin{proof}
We will start by proving the second part of the theorem, namely that alternate rounding and complementing produces a basic maximum-weight fractional matching with lesser odd cycles. For Case (i), let $\bar{x}$ be the basic fractional matching obtained by alternate rounding on $C_i$ at $v$. Since $y_v=0$, both $\bar{x}$ and $y$ satisfy complementary slackness. Hence, $\bar{x}$ is optimal to (P) and $\mathscr{C}(\bar{x})=\mathscr{C}(\hat{x})\setminus C_i$. For Case (ii), denote $u=V(P)\cap V(C_i)$ and $v=V(P)\cap V(C_j)$ as the endpoints of $P$. Let $\bar{x}$ be the basic fractional matching obtained by alternate rounding on $C_i,C_j$ at $u,v$ respectively and complementing on $P$. Note that $u$ and $v$ are exposed after the alternate rounding, and covered after complementing. Since $\bar{x}$ and $y$ satisfy complementary slackness, $\bar{x}$ is optimal to (P) and $\mathscr{C}(\bar{x})=\mathscr{C}(\hat{x})\setminus\set{C_i,C_j}$. For Case (iii), denote $u=V(P)\cap V(C_i)$ and $v\notin V(\mathscr{C}(\hat{x}))$ as the endpoints of $P$. Let $\bar{x}$ be the basic fractional matching obtained by alternate rounding on $C_i$ at $u$ and complementing on $P$. Since $y_v=0$, both $\bar{x}$ and $y$ satisfy complementary slackness. Thus, $\bar{x}$ is optimal to (P) and $\mathscr{C}(\bar{x})=\mathscr{C}(\hat{x})\setminus C_i$.

Next, we prove the first part of the theorem. We may assume $y_v>0$ for every vertex $v\in V(\mathscr{C}(\hat{x}))$. Let $x^*$ be a basic maximum-weight fractional matching in $G$ such that $\size{\mathscr{C}(x^*)}=\gamma(G)$. Define $N(\hat{x}) := M(\hat{x})\setminus E(\mathscr{C}(x^*))$ and $N(x^*) := M(x^*)\setminus E(\mathscr{C}(\hat{x}))$. Consider the following subgraph 
\[J = (V,N(\hat{x}) \triangle N(x^*)).\]
Since $N(\hat{x})$ and $N(x^*)$ are matchings in $G$, $J$ is made up of vertex-disjoint paths and cycles of $G$. For each such path or cycle, its edges alternately belong to $N(\hat{x})$ or $N(x^*)$. Moreover, its intermediate vertices are disjoint from $\mathscr{C}(\hat{x})$ and $\mathscr{C}(x^*)$. Since $\hat{x}$ and $x^*$ are maximum-weight fractional matchings in $G$, every path in $J$ is tight by complementary slackness. If there exists a path in $J$ which connects two odd cycles from $\mathscr{C}(\hat{x})$, then we are done. If there exists a path in $J$ which connects an odd cycle from $\mathscr{C}(\hat{x})$ and a vertex $v\notin V(\mathscr{C}(\hat{x})\cup\mathscr{C}(x^*))$, then $y_v=0$ because $v$ is either exposed by $M(\hat{x})$ or $M(x^*)$. Hence, we are also done. So we may assume every path in $J$ belongs to one of the following three categories:
  \begin{enumerate}[noitemsep,topsep=0pt]
    \item[(a)] Vertex disjoint from $\mathscr{C}(\hat{x})$ and $\mathscr{C}(x^*)$.
    \item[(b)] Starts and ends at the same cycle.
    \item[(c)] Connects an odd cycle from $\mathscr{C}(\hat{x})$ and an odd cycle from $\mathscr{C}(x^*)$.
  \end{enumerate}
Note that by the second part of the theorem, there is no path in $J$ which connects two odd cycles from $\mathscr{C}(x^*)$ or an odd cycle from $\mathscr{C}(x^*)$ and a vertex $v\notin V(\mathscr{C}(\hat{x})\cup \mathscr{C}(x^*))$. We say that two odd cycles $C_i$ and $C_j$ are \emph{adjacent} if $V(C_i)\cap V(C_j)\neq \emptyset$ or if they are connected by a path in $J$. 

\begin{claim}
  Every cycle in $\mathscr{C}(\hat{x})$ is adjacent to a cycle in $\mathscr{C}(x^*)$.
\end{claim}

\begin{proof}
  Let $C$ be an odd cycle in $\mathscr{C}(\hat{x})$. For every vertex $v\in V(C)$, since we assumed $y_v>0$, by complementary slackness it is either $M(x^*)$-covered or belongs to $V(\mathscr{C}(x^*))$. If $v\in V(\mathscr{C}(x^*))$, then we are done. So we may assume that every vertex in $C$ is $M(x^*)$-covered. Let $uv\in M(x^*)$ where $u\in V(C)$ and $v\notin V(C)$. Observe that $uv$ is the first edge of a path in $J$, so it either ends at an odd cycle in $\mathscr{C}(x^*)$ or $C$. Since $C$ has an odd number of vertices, by the pigeonhole principle there exists a path in $J$ which connects $C$ and an odd cycle in $\mathscr{C}(x^*)$.
\end{proof}

Recall that we assumed no two cycles in $\mathscr{C}(\hat{x})$ are adjacent. We also know that no two cycles in $\mathscr{C}(x^*)$ are adjacent. Since $\size{\mathscr{C}(\hat{x})}>\size{\mathscr{C}(x^*)}$, by the previous claim there exists an odd cycle in $\mathscr{C}(x^*)$ which is adjacent to at least two odd cycles in $\mathcal{C}(\hat{x})$. Let $C^*\in \mathscr{C}(x^*)$ be adjacent to $C_1,\dots,C_k\in \mathscr{C}(\hat{x})$ for some $k\geq 2$. For every $i\in[k]$, define
\[S_i:=\set{v\in V(C^*):v\in V(C_i) \text{ or } \exists\; \text{a path in }J \text{ from }v \text{ to }C_i}\]
and $S_0:=V(C^*)\setminus \cup_{i=1}^kS_i$. Note that $y_v>0$ for every vertex $v\in V(C^*)$. Hence, by complementary slackness every vertex in $S_0$ is $M(\hat{x})$-covered. Let $v\in S_0$. It is either matched to another vertex in $S_0$ or is an endpoint of a path in $J$ whose other endpoint is also a vertex in $S_0$. Hence, $\size{S_0}$ is even. Moreover, $S_i\neq \emptyset$ for all $i\geq 1$, and the sets $S_0,\dots,S_k$ partition $V(C^*)$. Let $\mathcal{P}$ be the set of paths in $J$ that start and end at $C^*$, and consider the subgraph $C^*\cup \mathcal{P}$. We claim that there exists an $M(\hat{x})$-alternating path from $S_i$ to $S_j$ in $C^*\cup \mathcal{P}$ for some $j\neq i$. Since every path in $\mathcal{P}$ starts and ends with an edge in $M(\hat{x})$, we can perform the following reduction: contract every path in $\mathcal{P}$ into a single edge in $M(\hat{x})$. It is easy to see that an $M(\hat{x})$-alternating path from $S_i$ to $S_j$ in $C^*\cup \mathcal{P}$ corresponds to an $M(\hat{x})$-alternating path from $S_i$ to $S_j$ in the reduced graph. Then, observe that the reduced graph along with the matching $M(\hat{x})$ forms a connector. By Lemma \ref{lem:connector}, there exists an $M(\hat{x})$-alternating path $P$ from $S_i$ to $S_j$ in $C^*\cup \mathcal{P}$ for some $j\neq i$.

Let $v_i\in S_i$ and $v_j\in S_j$ be the endpoints of $P$. Let $P_i$ and $P_j$ be the paths in $J$ connecting $v_i$ to $C_i$ and $v_j$ to $C_j$ respectively. If $v_i\in V(C_i)$, set $P_i=\emptyset$. Similarly if $v_j\in V(C_j)$, set $P_j=\emptyset$. Then, $P_i\cup P\cup P_j$ forms a tight $M(\hat{x})$-alternating path which connects $C_i$ and $C_j$.
\end{proof}

%\begin{corollary}
%Let $\hat{x}$ be a basic maximum-weight fractional matching in $G$. Then, there exists a basic maximum-weight fractional matching $x^*$ in $G$ such that $\size{\mathscr{C}(x^*)}=\gamma(G)$ and $\mathscr{C}(x^*)\subseteq \mathscr{C}(\hat{x})$.
%\end{corollary}

Given a basic maximum-weight fractional matching $\hat{x}$ in $G$, we would like to reduce the number of odd cycles in $\mathscr{C}(\hat{x})$ to $\gamma(G)$. One way to accomplish this is to search for the structures described in Theorem \ref{thm:cycles}. Fix a minimum fractional $w$-vertex cover $y$ in $G$. Let $G'$ be the unweighted graph obtained by applying the following operations to $G$ (see Figure \ref{fig:auxgraph}): 

\begin{enumerate}[noitemsep,topsep=0pt]
  \item[(a)] Delete all non-tight edges.
  \item[(b)] Add a vertex $z$. 
  \item[(c)] For every vertex $v\in V$ where $\hat{x}(\delta(v))=1$ and $y_v=0$, add the edge $vz$.
  \item[(d)] For every vertex $v\in V$ where $\hat{x}(\delta(v))=0$ and $y_v=0$, add the vertex $v'$ and the edges $vv',v'z$.
  \item[(e)] Shrink every odd cycle $C_i\in \mathscr{C}(\hat{x})$ into a pseudonode $i$.
\end{enumerate}

Note that none of the edges in $M(\hat{x})$ and $\mathscr{C}(\hat{x})$ were deleted because they are tight. Consider the edge set $M':=M(\hat{x})\cup \set{vv':v\in V}$. It is easy to see that $M'$ is a matching in $G'$. The significance of the auxiliary graph $G'$ is given by the following lemma:

\begin{figure}[ht]
\centering
\def\dist{1}
\begin{tikzpicture}[node distance=\dist cm, inner sep=2.5pt, minimum size=2.5pt, auto]
  % Nodes
  \node [vertex] (z) at (0,4) {z};
  \node [vertex] (m1) at (1,1) {};
  \node [vertex] (m2) [above right of=m1] {};
  \node [vertex] (m3) at (3,1.5) {};
  \node [vertex] (m4) [right of=m3] {};
  \node [vertex,fill=gray!70] (p1) at (-0.5,1) {};
  \node [vertex,fill=gray!70] (p2) at (-1.5,1.5) {};
  \node [vertex,fill=gray!70] (p3) at (0.3,2) {};
  \node [vertex] (u1) at (-4,1.5) {};
  \node [vertex] (v1) at (-4,3) {};
  \node [vertex] (u2) at (-2.5,1.5) {};
  \node [vertex] (v2) at (-2.5,3) {};
  
  % Graph G
  \draw (0,1.5) ellipse (5.5cm and 1cm);  

  %plot[smooth, tension=0.7] coordinates {(-1.5,0) (-2,0.2) (-3.5,0.6) (-5.5,1.5) (0,2.5) (5.5,1.5) (3.5,0.5) (0,0) (-1.5,0)};
  % Edges
  \path [matched edge] (m1) -- (m2);
  \path [matched edge] (m3) -- (m4);
  \path [matched edge] (u1) -- (v1);
  \path [matched edge] (u2) -- (v2);
  \path (z) edge (m2);
  \path (z) edge (v1);
  \path (z) edge (v2);
  \path (z) edge (p1);
\end{tikzpicture}
\caption{The auxiliary graph $G'$ and the matching $M'$. Vertices in the ellipse are from the original graph $G$. Gray vertices represent pseudonodes.}
\label{fig:auxgraph}
\end{figure}
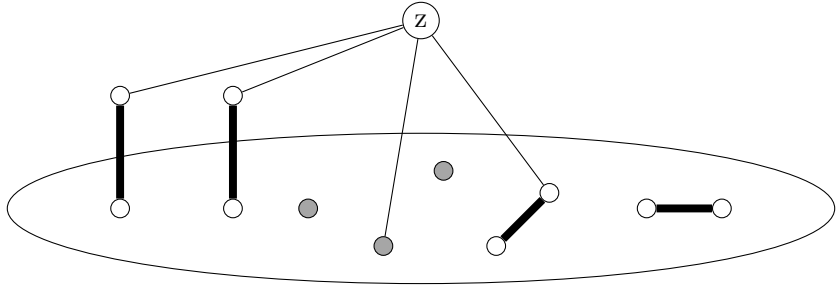

\begin{lemma}
\label{lem:augpaths}
  $M'$ is a maximum matching in $G'$ if and only if $\size{\mathscr{C}(\hat{x})}=\gamma(G)$.
\end{lemma}

\begin{proof}
$(\Rightarrow)$ Let $\hat{x}$ be a basic maximum-weight fractional matching where $\size{\mathscr{C}(\hat{x})}>\gamma(G)$ and $y$ be a minimum fractional $w$-vertex cover in $G$. Applying Theorem \ref{thm:cycles} yields three cases. In Case (i), there exists a vertex $v\in C_i$ for some odd cycle $C_i\in \mathscr{C}(\hat{x})$ such that $y_v=0$. Then, the edge $iz$ is an $M'$-augmenting path in $G'$. In Case (ii), there exists a tight $M(\hat{x})$-alternating path $P$ in $G$ connecting two odd cycles $C_i,C_j\in \mathscr{C}(\hat{x})$. In $G'$, $P$ is an $M'$-augmenting path whose endpoints are pseudonodes $i$ and $j$. In Case (iii), there exists a tight and valid $M(\hat{x})$-alternating path $P$ in $G$ connecting an odd cycle $C_i\in\mathscr{C}(\hat{x})$ and a vertex $v\notin V(\mathscr{C}(\hat{x}))$ such that $y_v=0$. If $v$ is $M(\hat{x})$-covered, then $P+vz$ is an $M'$-augmenting path in $G'$. Otherwise, $P+vv'+v'z$ is an $M'$-augmenting path in $G'$. Thus, $M'$ is not a maximum matching in $G'$.

$(\Leftarrow)$ Assume $M'$ is not a maximum matching in $G'$. Then, there exists an $M'$-augmenting path $P$ in $G'$. If both of its endpoints are pseudonodes $i$ and $j$, then $P$ is a tight $M(\hat{x})$-alternating path in $G$ which connects $C_i$ and $C_j$. So we may assume the endpoints of $P$ are a pseudonode $i$ and $z$. If $iz\in E(P)$, then there exists a vertex $v\in V(C_i)$ such that $y_v=0$. If $vz\in E(P)$ for some $v\in V$, then $y_v=0$ and $v$ is $M(\hat{x})$-covered. Hence, $P-vz$ is a tight and valid $M(\hat{x})$-alternating path in $G$ connecting $C_i$ and $v$. Otherwise, $v'z\in E(P)$ for some $v\in V$, which implies that $y_v=0$ and $v$ is $M(\hat{x})$-exposed. Hence, $P-vv'-v'z$ is a tight and valid $M(\hat{x})$-alternating path in $G$ connecting $C_i$ and $v$. By Theorem \ref{thm:cycles}, $\size{\mathscr{C}(\hat{x})}>\gamma(G)$.
\end{proof}

Thus, searching for the structures in Theorem \ref{thm:cycles} is equivalent to searching for an $M'$-augmenting path in $G'$. This immediately gives us an algorithm to generate a basic maximum-weight fractional matching with $\gamma(G)$ odd cycles.

\begin{algorithm}[H]
	\caption{Minimize number of odd cycles}
	\label{alg:mincycle}

  Compute a basic maximum-weight fractional matching $\hat{x}$ in $G$\;
  Compute a minimum fractional $w$-vertex cover $y$ in $G$\;
  Construct $G'$ and $M'$\;
	
	\While{$\exists$ an $M'$-exposed pseudonode $r$ in $G'$}{
		Grow an $M'$-alternating tree $T$ rooted at $r$ using Edmonds' algorithm \cite{journals/cjm/Edmonds65}\;
		\uIf{an $M'$-augmenting $r$-$s$ path $P'$ is found in $G'$}{
			Let $P$ be the corresponding tight $M(\hat{x})$-alternating path in $G$\;
     	\uIf{$s$ is a pseudonode}{
        	Alternate round on $C_r,C_s$ and complement on $P$\;
      	}
      	\Else{
			Alternate round on $C_r$ and complement on $P$\;      		
		}
      	Update $G'$ and $M'$
		}
		\Else{
			$G'\leftarrow G'\setminus V(T)$
		}
	}
	\Return $\hat{x}$
\end{algorithm}

After an $M'$-augmenting path $P'$ is found, let $\bar{x}$ denote the new basic maximum-weight fractional matching in $G$ obtained by alternate rounding and complementing $\hat{x}$. We can update $G'$ as follows. If $s$ is a pseudonode, we unshrink $C_r$ and $C_s$ in $G'$ because $\mathscr{C}(\bar{x})=\mathscr{C}(\hat{x})\setminus\set{C_r,C_s}$. Otherwise, $s=z$ and we only unshrink $C_r$. Then, there are two cases. In the first case, we have $vz\in E(P')$ for some $v\in V$. Observe that $\hat{x}(\delta(v))=1$ but $\bar{x}(\delta(v))=0$. Hence we replace the edge $vz$ with edges $vv',v'z$. In the second case, we have $v'z\in E(P')$ for some $v\in V$. This implies $\hat{x}(\delta(v))=0$ but $\bar{x}(\delta(v))=1$. So we replace edges $vv',v'z$ with the edge $vz$.

\begin{theorem}
\label{thm:mincycles}
Algorithm \ref{alg:mincycle} computes a basic maximum-weight fractional matching with $\gamma(G)$ odd cycles in polynomial time.
\end{theorem}

\begin{proof}
There are at most $O(n)$ vertex-disjoint odd cycles in $\mathscr{C}(\hat{x})$. At every iteration, we eliminate at least one odd cycle from $\mathscr{C}(\hat{x})$ or a frustrated tree from $G'$. Hence, there are at most $O(n)$ iterations, and Algorithm \ref{alg:mincycle} terminates in polynomial time. Next, we prove correctness. Suppose we obtain an $M'$-frustrated tree $T$. Every edge in $T$ has one endpoint in $A(T)$ and another endpoint in $B(T)$. Every edge in $\delta_{G'}(T)$ has one endpoint in $A(T)$ and another endpoint outside $T$. Since the matching in $T$ remains unchanged in every iteration, this property continues to hold throughout the execution of the algorithm. Thus, $T$ is a frustrated tree in every subsequent iteration. This implies that the last matching generated by the algorithm is maximum. By Lemma \ref{lem:augpaths}, we have $\size{\mathscr{C}(\hat{x})}=\gamma(G)$.
\end{proof}

We remark here that in Algorithm \ref{alg:mincycle}, we can avoid solving linear programs to obtain $\hat{x}$ and $y$ in Steps 1 and 2. They can be computed using a simple duplication technique by Nemhauser and Trotter~\cite{journals/mp/NemhauserT75}, which involves solving the problem on a suitable bipartite graph.

\section{Computing vertex-stabilizers}
\label{sec:vertex_stabilizer}

The goal of this section is to prove the following theorem:

\begin{theorem}
\label{thm:vert_stabilizer}
There exists a polynomial-time algorithm that computes a minimum vertex-stabilizer $S$ for a weighted graph $G$. 
Moreover, $\nu(G\setminus S)\geq \frac{2}{3}\nu(G)$.
\end{theorem}

\noindent
Let us start with discussing a lower bound on the size of a minimum vertex-stabilizer.

\paragraph{Lower bound.} We will here prove that $\gamma(G)$ is a lower bound on the number of vertices to remove in order to stabilize a graph. Recall that a graph is stable if and only if $\gamma(G)=0$. One strategy to achieve this is by showing that $\gamma(G)$ does not decrease by too much when we remove a vertex. Indeed, we prove that $\gamma(G)$ drops by at most 1 when a vertex is deleted (Lemma \ref{lem:gamma_vert}). We first develop a couple of claims.

\begin{claim}
\label{clm:frac}
Let $\hat{x}$ and $y$ be a basic maximum-weight fractional matching and a minimum fractional $w$-vertex cover in $G$ respectively. Pick a vertex $s$ from any odd cycle $C\in \mathscr{C}(\hat{x})$. If $\bar{x}$ is the fractional matching obtained by alternate rounding on $C$ at $s$, then $\bar{x}_{-\delta(s)}$ and $y_{-s}$ is a basic maximum-weight fractional matching and a minimum fractional $w$-vertex cover in $G\setminus s$ respectively.
\end{claim}

\begin{proof}
First, notice that $\bar{x}_{-\delta(s)}$ is a basic fractional matching and $y_{-s}$ is a fractional $w$-vertex cover in $G\setminus s$. We will show that they satisfy complementary slackness. Let $uv\in E(C)$ be an edge where $\bar{x}_{uv}>0$. Since $e\in E(C)$, we have $\hat{x}_{uv}>0$ and so $y_u+y_v=w_{uv}$. Next, let $v\neq s$ be a vertex in $C$ where $y_v>0$. We only need to check the vertices in $C$ because $\hat{x}_e=0$ for every edge $e\in\delta(s)\setminus E(C)$. Since $v$ is $M(\bar{x})$-covered, we have $\bar{x}(\delta(v))=1$. Therefore, $\bar{x}_{-\delta(s)}$ and $y_{-s}$ form a primal-dual optimal pair.
\end{proof}

The following operation allows us to switch between fractional matchings on a set of edges:

\begin{definition}
Let $x$ and $x'$ be fractional matchings in $G$. By \emph{switching} on $E'\subseteq E$ from $x$ to $x'$, we mean replacing $x_e$ by $x'_e$ for all $e\in E'$.
\end{definition}

Switching does not necessarily yield a feasible fractional matching. Hence, we will only use it on the components of a specific subgraph of $G$:

\begin{claim}
\label{clm:switch}
Given two basic fractional matchings $x$ and $x'$, let $H$ be the subgraph of $G$ induced by $\supp(x+x')$. For any component $K$ in $H$, switching on $E(K)$ from $x$ to $x'$ yields a basic fractional matching in $G$.
\end{claim}

\begin{proof}
Let $\bar{x}$ denote the vector obtained by switching on $E(K)$ from $x$ to $x'$. We first show that $\bar{x}$ is a feasible fractional matching in $G$. For the purpose of contradiction, suppose there exists a vertex $v\in V(K)$ such that $\bar{x}(\delta(v))>1$. Since $\bar{x}_e = x'_e$ for all $e\in E(K)$ and $\bar{x}_e=x_e$ for all $e\notin E(K)$, we have $0<\bar{x}(\delta(v)\cap E(K)) \leq 1$ and $0<\bar{x}(\delta(v)\setminus E(K)) \leq 1$. So there exists an edge $f\in \delta(v)\setminus E(K)$ such that $x_f>0$, which is a contradiction. It is easy to see that $\bar{x}$ is basic.
\end{proof}

\begin{lemma}\label{lem:gamma_vert}
For every vertex $v\in V$, $\gamma(G\setminus v)\geq \gamma(G)-1$.
\end{lemma}

\begin{proof}
Let $x^*$ be a basic maximum-weight fractional matching in $G$ such that $\size{\mathscr{C}(x^*)}=\gamma(G)$. Let $y$ be a minimum fractional $w$-vertex cover in $G$. For the purpose of contradiction, suppose there exists a vertex $u\in V$ such that $\gamma(G\setminus u)<\gamma(G)-1$. There are two cases:

\smallskip
\noindent
\emph{Case 1: $u\in V(C)$ for some odd cycle $C\in \mathscr{C}(x^*)$.}
Let $\bar{x}$ be the fractional matching obtained from $x^*$ by alternate rounding on $C$ at $u$. By Claim \ref{clm:frac}, we know that $\bar{x}_{-\delta(u)}$ is a basic maximum-weight fractional matching and $y_{-u}$ is a minimum fractional $w$-vertex cover in $G\setminus u$. We first give a proof sketch for this case. If $\bar{x}_{-\delta(u)}$ is not an optimal basic solution yielding $\gamma(G\setminus u)$ odd cycles, then one of the structures given by Theorem \ref{thm:cycles} must exist. This same structure would be a structure corresponding to the basic solution $x^*$, but this yields a contradiction since $x^*$ is an optimal basic solution with $\gamma(G)$ odd cycles.

For notational convenience, we can use $\mathscr{C}(\bar{x})$ and $M(\bar{x})$ to refer to the odd cycles and matched edges of $\bar{x}_{-\delta(u)}$ respectively because $\mathscr{C}(\bar{x})=\mathscr{C}\pr{\bar{x}_{-\delta(u)}}$ and $M(\bar{x})=M\pr{\bar{x}_{-\delta(u)}}$. Since $\size{\mathscr{C}(\bar{x})}=\size{\mathscr{C}(x^*)}-1=\gamma(G)-1>\gamma(G\setminus u)$, Theorem \ref{thm:cycles} tells us that $G\setminus u$ contains one of the following structures. The first structure is a vertex $v\in V(C_i)$ for some odd cycle $C_i\in \mathscr{C}(\bar{x})$ such that $y_v=0$. However, since $C_i\in \mathscr{C}(x^*)$, by Theorem \ref{thm:cycles} we arrive at the contradiction $\size{\mathscr{C}(x^*)}>\gamma(G)$. The second structure is a tight and valid $M(\bar{x})$-alternating path $P$ which connects two odd cycles $C_i,C_j\in\mathscr{C}(\bar{x})$, or an odd cycle $C_i\in \mathscr{C}(\bar{x})$ and a vertex $v\notin V(\mathscr{C}(\bar{x}))$ such that $y_v=0$. Note that $C_i,C_j\in \mathscr{C}(x^*)$. If $V(P)\cap V(C)=\emptyset$, then $P$ is also a tight and valid $M(x^*)$-alternating path in $G$ which connects $C_i$ and $C_j$, or $C_i$ and $v$. So, let $s=V(C_i)\cap V(P)$ and $t$ denote the first vertex of $C$ encountered while traversing along $P$ from $s$. Then, the $s$-$t$ subpath of $P$ is a tight $M(x^*)$-alternating path which connects $C_i,C\in \mathscr{C}(x^*)$. We again obtain the contradiction $\size{\mathscr{C}(x^*)}>\gamma(G)$ by Theorem \ref{thm:cycles}.

\smallskip
\noindent
\emph{Case 2: $u\notin V(\mathscr{C}(x^*))$.}
If $u$ is $M(x^*)$-exposed, then $\nu_f(G\setminus u)=\nu_f(G)$ and $\gamma(G\setminus u)=\gamma(G)$. So we may assume $u$ is $M(x^*)$-covered. Let $\hat{x}$ be a basic maximum-weight fractional matching in $G\setminus u$ such that $\size{\mathscr{C}(\hat{x})}< \gamma(G)-1$. Define $N(\hat{x}):=M(\hat{x})\setminus E(\mathscr{C}(x^*))$ and $N(x^*):=M(x^*)\setminus E(\mathscr{C}(\hat{x}))$. Consider the subgraph $J=(V,N(x^*)\triangle N(\hat{x}))$. Note that $u$ is covered by $N(x^*)$ and exposed by $N(\hat{x})$. Let $P$ be the component in $J$ which contains $u$. We know that $P$ is a path with $u$ as an endpoint. Let $v$ be the other endpoint of $P$. There are 3 subcases, but before jumping into them, we first give an overview of how we arrive at a contradiction in each subcase. We show that one can move from $x^*$ to a new solution $\tilde x$ such that:
\begin{enumerate}[noitemsep,topsep=0pt]
  \item[(i)] $\tilde{x}$ is a basic maximum-weight fractional matching for a subgraph $G'$ obtained by deleting at most 1 vertex from a cycle of $\mathscr{C}(x^*)$; and
  \item[(ii)] $\size{\mathscr{C}(\tilde{x})}<\gamma(G')$.
\end{enumerate}
Clearly, both of the above properties cannot hold, so this yields a contradiction.

\smallskip
\emph{Subcase 2.1: $v\in C$ for some odd cycle $C\in \mathscr{C}(x^*)$.}
In this subcase, the path $P$ has even length. Let $\bar{x}$ be the fractional matching obtained from $x^*$ by alternate rounding on $C$ at $v$. By Claim \ref{clm:frac}, $\bar{x}_{-\delta(v)}$ is a basic maximum-weight fractional matching in $G\setminus v$. Let $H$ be the subgraph of $G$ induced by $\supp(\hat{x}+\bar{x})$ (see Figure \ref{fig:subcase1} for an example). Note that $\hat{x}_e+\bar{x}_e=0$ for every edge $e\notin E(P)$ which is incident to a vertex in $P$. Thus, $P$ is a component in $H$. Since $\size{\mathscr{C}(\bar{x})}=\gamma(G)-1>\size{\mathscr{C}(\hat{x})}$, there exists a component $K$ in $H$ which has more odd cycles from $\mathscr{C}(\bar{x})$ than $\mathscr{C}(\hat{x})$. Switching on $K$ from $\bar{x}_{-\delta(v)}$ to $\hat{x}$ yields a basic fractional matching in $G\setminus v$ with less than $\gamma(G)-1$ odd cycles. To yield a contradiction to Case 1, it is left to show that it is maximum-weight. This is because we are deleting a vertex $v$ from an odd cycle of $\mathscr{C}(x^*)$, but $\gamma(G\setminus v)$ decreases by more than 1. Now, since $\hat{x}$ and $\bar{x}_{-\delta(v)}$ are maximum-weight fractional matchings in $G\setminus u$ and $G\setminus v$ respectively, we have $\sum_{e\in E(K)}w_e\hat{x}_e = \sum_{e\in E(K)}w_e\bar{x}_e$ because $u,v\notin V(K)$. Thus, the resulting matching is indeed maximum-weight in $G\setminus v$.

\begin{figure}[ht]
\centering
\def\dist{1.3}
\begin{tikzpicture}[node distance=\dist cm, inner sep=2.5pt, minimum size=2.5pt, auto]
  % Nodes
  \node [vertex,label=below:\small{$v$}] (v1) {};
  \node [vertex] (v2) [above right of=v1] {};
  \node [vertex] (v3) [below right of=v1] {};
  \node [vertex] (v4) [right of=v2] {};
  \node [vertex] (v5) [right of=v3] {};
  \node [vertex] (v6) [left of=v1] {};
  \node [vertex] (v7) [left of=v6] {};
  \node [vertex] (v8) [left of=v7] {};
  \node [vertex,label=below:\small{$u$}] (v9) [left of=v8] {};

  % Edges
  \path [matched edge] (v2) -- (v4);
  \path [matched edge] (v3) -- (v5);  
  \path [dashed edge] (v1) -- (v2);
  \path [dashed edge] (v1) -- (v3);
  \path [dashed edge] (v4) -- (v5);
  \path [selected edge] (v1) -- (v6);
  \path [matched edge] (v6) -- (v7);
  \path [selected edge] (v7) -- (v8);
  \path [matched edge] (v8) -- (v9);

  % Label
  \node[weight] (c) [right of=v1] {$C$};
  \node[weight] (k) at (1.25,1.75) {$K$};

  % Component
  \draw plot [smooth cycle] coordinates {(0,1) (-1,2) (2,2.5) (3,1) (2,0.5)};
\end{tikzpicture}
\caption{An example of the graph induced by $\supp(\hat{x}+\bar{x})$ in Subcase 2.1. Black bold edges are in $M(\bar{x})$ while gray bold edges are in $M(\hat{x})$.}
\label{fig:subcase1}
\end{figure}
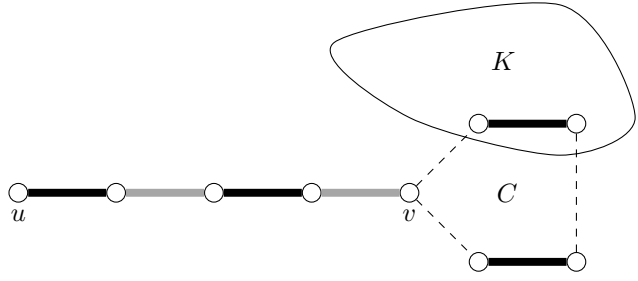

\smallskip
\emph{Subcase 2.2: $v\in C$ for some odd cycle $C\in \mathscr{C}(\hat{x})$.}
In this subcase, the path $P$ has odd length. Let $\bar{x}$ be the fractional matching obtained from $\hat{x}$ by alternate rounding on $C$ at $v$. By Claim \ref{clm:frac}, $\bar{x}_{-\delta(v)}$ is a basic maximum-weight fractional matching in $G\setminus\set{u,v}$. Let $H$ be the subgraph of $G$ induced by $\supp(x^*+\bar{x})$ (see Figure \ref{fig:subcase2} for an example). Note that $x^*_e+\bar{x}_e=0$ for every edge $e\notin E(P)$ incident to a vertex in $P$. Thus, $P$ is a component in $H$. Since $\size{\mathscr{C}(\bar{x})} = \size{\mathscr{C}(\hat{x})} - 1 < \gamma(G)-2<\size{\mathscr{C}(x^*)}$, there exists a component $K$ in $H$ which has more odd cycles from $\mathscr{C}(x^*)$ than $\mathscr{C}(\bar{x})$. Switching on $K$ from $x^*$ to $\bar{x}$ yields a basic fractional matching in $G$ with less than $\gamma(G)$ odd cycles. To yield a contradiction, it is left to show that it is maximum-weight. Since $x^*$ and $\bar{x}_{-\delta(v)}$ are maximum-weight fractional matchings in $G$ and $G\setminus\set{u,v}$ respectively, we have $\sum_{e\in E(K)}w_ex^*_e=\sum_{e\in E(K)}w_e\bar{x}_e$ because $u,v\notin V(K)$. Thus, the resulting basic fractional matching is maximum-weight in $G$.

\begin{figure}[ht]
\centering
\def\dist{1.3}
\begin{tikzpicture}[node distance=\dist cm, inner sep=2.5pt, minimum size=2.5pt, auto]
  % Nodes
  \node [vertex,label=below:\small{$v$}] (v1) {};
  \node [vertex] (v2) [above right of=v1] {};
  \node [vertex] (v3) [below right of=v1] {};
  \node [vertex] (v4) [right of=v2] {};
  \node [vertex] (v5) [right of=v3] {};
  \node [vertex] (v6) [left of=v1] {};
  \node [vertex] (v7) [left of=v6] {};
  \node [vertex,label=below:\small{$u$}] (v8) [left of=v7] {};

  % Edges
  \path [selected edge] (v2) -- (v4);
  \path [selected edge] (v3) -- (v5);  
  \path [dashed edge] (v1) -- (v2);
  \path [dashed edge] (v1) -- (v3);
  \path [dashed edge] (v4) -- (v5);
  \path [matched edge] (v1) -- (v6);
  \path [selected edge] (v6) -- (v7);
  \path [matched edge] (v7) -- (v8);

  % Label
  \node[weight] (c) [right of=v1] {$C$};
  \node[weight] (k) at (1.25,1.75) {$K$};

  % Component
  \draw plot [smooth cycle] coordinates {(0,1) (-1,2) (2,2.5) (3,1) (2,0.5)};
\end{tikzpicture}
\caption{An example of the graph induced by $\supp(x^*+\bar{x})$ in Subcase 2.2. Black bold edges are in $M(x^*)$ while gray bold edges are in $M(\bar{x})$.}
\label{fig:subcase2}
\end{figure}
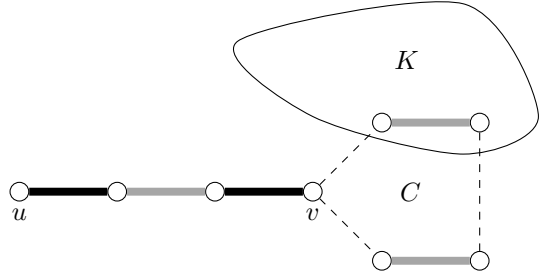

\smallskip
\emph{Subcase 2.3: $v\notin V(\mathscr{C}(x^*)\cup \mathscr{C}(\hat{x}))$.}
Let $H$ be the subgraph of $G$ induced by $\supp(x^*+\hat{x})$ (see Figure \ref{fig:subcase3} for an example). Note that $x^*_e+\hat{x}_e=0$ for every edge $e\notin E(P)$ which is incident to a vertex in $P$. Thus, the path $P$ is a component in $H$. Since $\size{\mathscr{C}(x^*)}>\gamma(G)-1>\size{\mathscr{C}(\hat{x})}$, there exists a component $K$ in $H$ which has more odd cycles from $\mathscr{C}(x^*)$ than $\mathscr{C}(\hat{x})$. Switching on $K$ from $x^*$ to $\hat{x}$ yields a basic fractional matching in $G$ with less than $\gamma(G)$ odd cycles. To yield a contradiction, it is left to show that it is maximum-weight. Since $x^*$ and $\hat{x}$ are maximum-weight fractional matchings in $G$ and $G\setminus u$ respectively, we have $\sum_{e\in E(K)}w_e x^*_e=\sum_{e\in E(K)}w_e\hat{x}_e$ because $u\notin V(K)$. This implies that the resulting basic fractional matching is maximum-weight in $G$.
\end{proof}

\begin{figure}[ht]
\centering
\def\dist{1.3}
\begin{tikzpicture}[node distance=\dist cm, inner sep=2.5pt, minimum size=2.5pt, auto]
  % Nodes
  \node [vertex,label=below:\small{$v$}] (v1) {};
  \node [vertex] (v2) [left of=v1] {};
  \node [vertex] (v3) [left of=v2] {};
  \node [vertex] (v4) [left of=v3] {};
  \node [vertex,label=below:\small{$u$}] (v5) [left of=v4] {};

  % Edges
  \path [selected edge] (v1) -- (v2);
  \path [matched edge] (v2) -- (v3);  
  \path [selected edge] (v3) -- (v4);
  \path [matched edge] (v4) -- (v5);

  % Label
  \node[weight] (k) at (1.25,1.75) {$K$};

  % Component
  \draw plot [smooth cycle] coordinates {(0,1) (-1,2) (2,2.5) (3,1) (2,0.5)};
\end{tikzpicture}
\caption{An example of the graph induced by $\supp(x^*+\hat{x})$ in Subcase 2.3. Black bold edges are in $M(x^*)$ while gray bold edges are in $M(\hat{x})$.}
\label{fig:subcase3}
\end{figure}
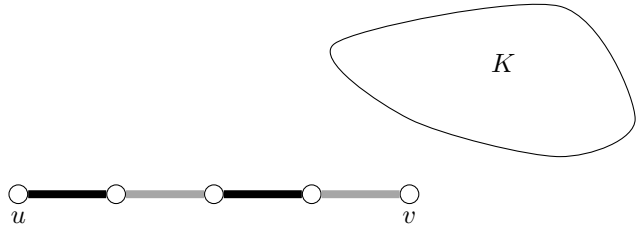

As a corollary to the above lemma, we obtain the claimed lower bound.

\begin{lemma}\label{lem:bound_vert}
For every vertex-stabilizer $S$ of $G$, $\size{S}\geq \gamma(G)$.
\end{lemma}

\paragraph{The algorithm.} The algorithm we use to stabilize a graph is very simple: it computes 
a basic maximum-weight fractional matching $\hat{x}$ in $G$ with $\gamma(G)$ odd cycles (this can be done using 
Algorithm \ref{alg:mincycle}) and 
a minimum fractional $w$-vertex cover $y$ in $G$, and then removes \emph{one} vertex from every cycle in $\mathscr{C}(\hat{x})$, namely, the vertex with the least $y$-value in the cycle. Algorithm \ref{alg:vert} formalizes this.

\begin{algorithm}[H]
Initialize $S \gets \emptyset$ \;
Compute a minimum fractional $w$-vertex cover $y$ in $G$\;
Compute a basic maximum-weight fractional matching $\hat{x}$ in $G$ with $\gamma(G)$ odd cycles\;
Let $\mathscr{C}(\hat{x})=\set{C_1,C_2,\dots,C_{\gamma(G)}}$\;
\For{$i = 1$ \textbf{to} $\gamma(G)$}{
  Let $v_i= \argmin_{v\in V(C_i)}y_v$\;
  $S\gets S+v_i$
}
\Return $S$ 
\caption{Minimum vertex-stabilizer}
\label{alg:vert}
\end{algorithm}

\noindent
We are now ready to prove the main theorem stated at the beginning of the section, Theorem \ref{thm:vert_stabilizer}.
\begin{proof}[Proof of Theorem \ref{thm:vert_stabilizer}.]
Let $S=\set{v_1,v_2,\dots,v_{\gamma(G)}}$ be the set of vertices returned by the algorithm. Let $\bar{x}$ be the vector obtained from $\hat{x}$ by alternate rounding on $C_i$ at $v_i$ for all $i$ respectively. By Lemma \ref{clm:frac}, $\bar{x}_{-\cup_{i=1}^{\gamma(G)}\delta(v_i)}$ is a basic maximum-weight fractional matching in $G\setminus S$. Note that it is also a maximum-weight integral matching in $G\setminus S$. Thus, $\nu(G\setminus S)=\nu_f(G\setminus S)$ and $G\setminus S$ is stable. Moreover, $S$ is minimum by Lemma \ref{lem:bound_vert}. It is left to show that $\nu(G\setminus S)\geq \frac{2}{3}\nu(G)$. For every odd cycle $C_i\in \mathscr{C}(\hat{x})$, we have 
\[y_{v_i}\leq \frac{y(V(C_i))}{\size{V(C_i)}}\leq \frac{y(V(C_i))}{3}\]
because $v_i$ has the smallest fractional $w$-vertex cover in $C_i$. From Lemma \ref{clm:frac}, we also know that $y_{-S}$ is a minimum fractional $w$-vertex cover in $G\setminus S$. Then,
\[\nu(G\setminus S) = \tau_f(G\setminus S) = \1^{\top}y-\sum_{i=1}^{\gamma(G)}y_{v_i}\geq \1^{\top}y - \frac{1}{3}\sum_{i=1}^{\gamma(G)}y(C_i) \geq \1^{\top}y - \frac{1}{3}\1^{\top}y = \frac{2}{3}\tau_f(G) \geq \frac{2}{3}\nu(G)\]
\end{proof}

Note that removing any single vertex from each cycle of $\mathscr{C}(\hat{x})$ yields a minimum-cardinality vertex stabilizer. The reason we chose the vertex with the smallest $y_v$ is to preserve the value of the original maximum-weight matching by a factor of $\frac{2}{3}$.

\paragraph{Tightness of the matching bound.} A natural question is whether it is possible to design an algorithm that always returns a vertex-stabilizer $S$ satisfying $\nu(G\setminus S) \geq \alpha \nu(G)$, for some $\alpha > \frac{2}{3}$. We report here an example showing that, in general, this is not possible since the bound of $\frac{2}{3}$ can be asymptotically tight. Consider the graph $G$ in Figure \ref{fig:twothirds} for some sufficiently small $\varepsilon>0$. It is unstable because it is an augmenting flower. The maximum-weight matching is given by the bold edges. For any vertex stabilizer $S$, 
\[\nu(G\setminus S)\leq 2 = \frac{2}{3-\varepsilon}\pr{3-\varepsilon} = \frac{2}{3-\varepsilon}\nu(G)\]

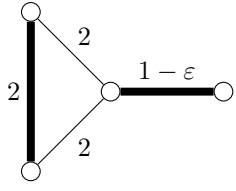
\begin{figure}[ht]
\centering
\def\dist{1.5}
\begin{tikzpicture}[node distance=\dist cm, inner sep=2.5pt, minimum size=2.5pt, auto]
  % Nodes
  \node [vertex] (v1) {};
  \node [vertex] (v2) [above left of=v1] {};
  \node [vertex] (v3) [below left of=v1] {};
  \node [vertex] (v4) [right of=v1] {};

  % Edges
  \path [matched edge] (v2) -- node [weight, left] {$2$} (v3);
  \path (v1) edge node [weight, above right] {$2$} (v2);
  \path [matched edge] (v1) -- node [weight, above] {$1-\varepsilon$} (v4);  
  \path (v1) edge node [weight, below right] {$2$} (v3);
\end{tikzpicture}
\caption{An example showing that the bound of $\frac{2}{3}$ is asymptotically tight.}
\label{fig:twothirds}
\end{figure}

Another natural question is whether one can at least distinguish if, for a specific instance, there exists
a vertex-stabilizer $S$ such that $\nu(G\setminus S)=\nu(G)$. Once again, we show that the answer is negative.
Specifically, let us call a vertex-stabilizer $S$ \emph{weight-preserving} if $\nu(G\setminus S)=\nu(G)$. We show that finding such a vertex-stabilizer is hard in general. The proof is based on a reduction from the independent set problem, similar to the one given by Bir\'{o} et al.~\cite{journals/tcs/BiroBGKP}.

\begin{theorem}
\label{thm:hardness_vert}
Deciding whether a graph has a weight-preserving vertex-stabilizer is NP-complete.
\end{theorem}

\begin{proof}
The problem is clearly in NP because any yes-instance can be verified using a weight-preserving vertex-stabilizer in polynomial time. To prove NP-hardness, we give a reduction from the independent set problem. Let $G=(V,E)$ and $k$ be an independent set instance, where $V=\set{v_1,v_2,\dots,v_n}$. The independent set problem asks to determine whether $G$ has an independent set of size at least $k$. We may assume $2\leq k\leq n$. We construct the gadget graph $G^*$ as follows. First, set the weight on every edge in $E$ to 1. For each $v_i\in V$, add a vertex $v'_i$ and the edge $v_iv'_i$ with weight 1. Denote this set of new vertices as $V'=\set{v'_1,v'_2,\dots,v'_n}$. Next, create $k$ pairwise-disjoint copies of the three cycle $C_i=(V_i,E_i)$ where $V_i=\set{a_i,b_i,c_i}$, $E_i=\set{a_ib_i,b_ic_i,a_ic_i}$ and the weight on every edge in $E_i$ is 4. Finally, add the edge $b_iv_j$ for every $i\in [k]$ and $j\in [n]$ with weight 2. (See Figure \ref{fig:hardness_vert})

\begin{figure}[ht]
\centering
\def\dist{1.5}
\def\gap{0.4}
\begin{tikzpicture}[node distance=\dist cm, inner sep=2.5pt, minimum size=2.5pt, auto]
% Nodes & Dots
\node [vertex,label=below left:\small{$v_1$}] (v1) {};
\node [vertex,label=below left:\small{$v_2$}] (v2) [right of=v1, xshift=\gap*\dist cm] {};
\node [font=\large] (d1) [right of=v2, xshift=\gap*\dist cm] {$\dots$};
\node [vertex,label=below left:\small{$v_k$}] (v3) [right of=d1, xshift=\gap*\dist cm] {};
\node [font=\large] (d2) [right of=v3, xshift=\gap*\dist cm] {$\dots$};
\node [font=\large] (d3) [above of=d1, yshift=\gap*\dist cm] {$\dots$};
\node [vertex,label=below left:\small{$v_n$}] (v4) [right of=d2, xshift=\gap*\dist cm] {};

\node [vertex,label=below left:\small{$v'_1$}] (w1) [below of=v1] {};
\node [vertex,label=below left:\small{$v'_2$}] (w2) [below of=v2] {};
\node [vertex,label=below left:\small{$v'_k$}] (w3) [below of=v3] {};
\node [vertex,label=below left:\small{$v'_n$}] (w4) [below of=v4] {};

\node [vertex,label=below left:\small{$b_1$}] (b1) [above of=v1] {};
\node [vertex,label=below left:\small{$b_2$}] (b2) [above of=v2] {};
\node [vertex,label=below left:\small{$b_k$}] (b3) [above of=v3] {};

\node [vertex] (a1) [above left of=b1, xshift=0.3*\dist cm] {};
\node [vertex] (a2) [above left of=b2, xshift=0.3*\dist cm] {};
\node [vertex] (a3) [above left of=b3, xshift=0.3*\dist cm] {};

\node [vertex] (c1) [above right of=b1, xshift=-0.3*\dist cm] {};
\node [vertex] (c2) [above right of=b2, xshift=-0.3*\dist cm] {};
\node [vertex] (c3) [above right of=b3, xshift=-0.3*\dist cm] {};

% Edges
\path (v1) edge (w1);
\path (v2) edge (w2);
\path (v3) edge (w3);
\path [matched edge] (v4) -- (w4);
\path (v3) edge [bend right=15] (v4);

\path [matched edge] (a1) -- (c1);
\path (b1) edge (a1);
\path (b1) edge (c1);
\path [matched edge] (b1) -- (v1);
\path (b1) edge (v2);
\path (b1) edge (v3);
\path (b1) edge (v4);

\path [matched edge] (a2) -- (c2);
\path (b2) edge (a2);
\path (b2) edge (c2);
\path (b2) edge (v1);
\path [matched edge] (b2) -- (v2);
\path (b2) edge (v3);
\path (b2) edge (v4);

\path [matched edge] (a3) -- (c3);
\path (b3) edge (a3);
\path (b3) edge (c3);
\path (b3) edge (v1);
\path (b3) edge (v2);
\path [matched edge] (b3) -- (v3);
\path (b3) edge (v4);

% Shaded Region
\path [fill=gray,opacity=0.25] (-\gap*\dist,-\gap*\dist)  -- (-\gap*\dist,\gap*\dist) -- (5*\dist+6*\gap*\dist,\gap*\dist) -- (5*\dist+6*\gap*\dist,-\gap*\dist) -- (-\gap*\dist,-\gap*\dist);

% Braces
\draw[decoration={brace,raise=15pt},decorate] (-\gap*\dist,\dist+0.05) -- node[weight,left=18pt] {$w_e=4$} (-\gap*\dist,1.75*\dist-0.05);
\draw[decoration={brace,raise=15pt},decorate] (-\gap*\dist,0.05) -- node[weight,left=18pt] {$w_e=2$} (-\gap*\dist,\dist-0.05);
\draw[decoration={brace,raise=15pt},decorate] (-\gap*\dist,-\dist+0.05) -- node[weight,left=18pt] {$w_e=1$} (-\gap*\dist,-0.05);
\end{tikzpicture}
\caption{The gadget graph $G^*$.}
\label{fig:hardness_vert}
\end{figure}
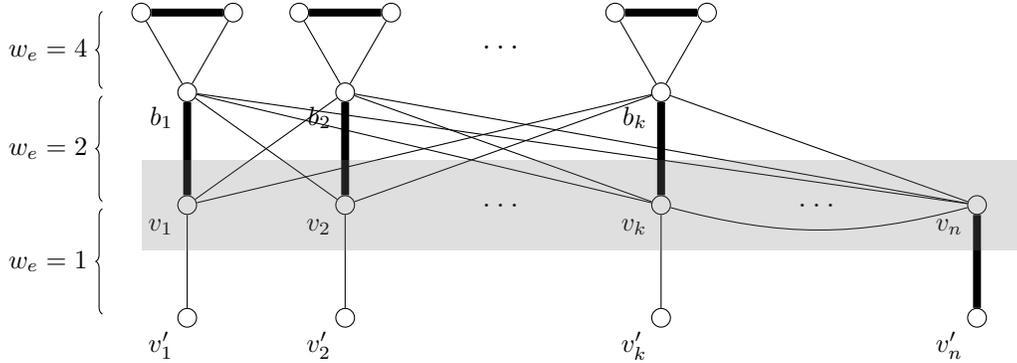

Our goal is to show that $G$ has an independent set of size at least $k$ if and only if $G^*$ has a weight-preserving vertex-stabilizer. Before proving the main result, we first derive some properties of maximum-weight matchings in $G^*$.

\begin{claim}\label{clm:matching1}
If $M$ is a maximum-weight matching in $G^*$, then $M\cap E=\emptyset$.
\end{claim}

\begin{proof}
For the purpose of contradiction, suppose there exists an edge $v_iv_j\in M\cap E$. Then, $(v'_i,v_i,v_j,v'_j)$ forms an augmenting path, which is a contradiction.
\end{proof}

\begin{claim}\label{clm:matching2}
If $M$ is a maximum-weight matching in $G^*$, every $b_i$ is matched to some $v_j$ in $M$.
\end{claim}

\begin{proof}
For the purpose of contradiction, suppose there exists an $i\in[k]$ such that $b_iv_j\notin M$ for all $j\in [n]$. Then, $b_i$ is either $M$-exposed or matched to $a_i$ or $c_i$. Since $k\leq n$, by the pigeonhole principle there exists an $\ell\in [n]$ such that $b_jv_\ell\notin M$ for all $j\in [k]$. By Claim \ref{clm:matching1}, $v_\ell v'_\ell\in M$. If $b_i$ is $M$-exposed, then $(b_i,v_\ell,v'_\ell)$ forms an augmenting path. Otherwise, we may without loss of generality assume $a_ib_i\in M$. Then, $(c_i,a_i,b_i,v_\ell,v'_\ell)$ forms an augmenting path. We have reached a contradiction.
\end{proof}

\begin{claim}\label{clm:matching3}
$\nu(G^*)=n+5k$
\end{claim}

\begin{proof}
Let $M$ be a maximum-weight matching in $G^*$. By Claim \ref{clm:matching2}, there are $k$ edges of the form $b_iv_j$ in $M$. Hence, there are also $k$ edges of the form $a_ic_i$ in $M$. Moreover, we have $n-k$ edges of the form $v_iv'_i$ in $M$. This gives a total weight of $2k+4k+n-k=n+5k$. 
\end{proof}

\begin{claim}\label{clm:matching4}
The set of inessential vertices in $G^*$ is $V'$.
\end{claim}

\begin{proof}
It is easy to see that the vertices in $G$ and $\cup_{i=1}^kC_i$ are essential because they are covered by every maximum-weight matching in $G^*$. Let $v'_i\in V'$ and $M$ be a maximum-weight matching in $G^*$ such that $v_iv'_i\in M$.
Since $k\geq 1$, there exist $j\in[k]$ and $\ell\in[n]$ such that $b_jv_\ell\in M$. Define a new matching $M':=M+b_jv_i-b_jv_\ell+v_\ell v'_\ell-v_iv'_i$. Note that $M'$ is a maximum-weight matching in $G^*$ and $v'_i$ is $M'$-exposed. Thus, $v'_i$ is inessential.
\end{proof}

For the forward direction, let $S$ be an independent set of $G$ where $\size{S}=k$. Without loss of generality, assume $S=\set{v_1,v_2,\dots,v_k}$. Let $M$ be the matching defined by
\[M:=\set{a_ic_i,b_iv_i:1\leq i\leq k}\cup\set{v_iv'_i:k< i\leq n}.\]
Since $w(M)=n+5k$, by Claim \ref{clm:matching3} it is a maximum-weight matching in $G^*$. We claim that $S':=\set{v'_1,v'_2,\dots, v'_k}$ is a weight-preserving vertex-stabilizer of $G^*$. First, note that the matching $M$ survives after removing $S'$ from $G^*$. Hence, $M$ is a maximum-weight matching in $G^*\setminus S'$. It is left to show that $G^*\setminus S'$ is stable. We define a fractional $w$-vertex cover $y$ on $G^*\setminus S'$ as follows:
\[y_v=\begin{cases}
   2, &\text{ if }v\in\cup_{i=1}^k\set{a_i,b_i,c_i}\\
   1, &\text{ if }v\in\cup_{i=k+1}^n v_i\\
   0, &\text{ otherwise}.
\end{cases}\]
It is easy to check that for every $uv\notin E$, the condition $y_u+y_v\geq w_{uv}$ holds. Let $v_iv_j\in E$. Since $S$ is an independent set, at most one of $v_i$ and $v_j$ is in $S$. This implies that $y_{v_i}+y_{v_j}\geq 1$. So $y$ is indeed a fractional $w$-vertex cover. Since $\tau_f(G^*\setminus S')=n+5k=\nu(G^*\setminus S')$, $G^*\setminus S'$ is stable.

For the converse, let $S'$ be a weight-preserving vertex-stabilizer of $G$. We know that $S'\subseteq V'$ by Claim \ref{clm:matching4} because $S'$ does not contain essential vertices. Let $M$ be a maximum-weight matching in $G^*\setminus S'$. Since $\nu(G^*\setminus S')=\nu(G^*)$, $M$ is also a maximum-weight matching in $G^*$. We claim that $\size{S'}=k$. If $\size{S'}<k$, then there exists an $i\in[n]$ such that $v_iv'_i\in G^*\setminus S'$ and $b_jv_i\in M$ for some $j\in[k]$. Then, $a_jc_j\in M$, so $C_j\cup (b_j,v_i,v'_i)$ forms an augmenting flower. This is a contradiction to Theorem \ref{thm:stable} because $G^*\setminus S'$ is stable. On the other hand, if $\size{S'}>k$, then $w(M)=6k+n-\size{S'}<n+5k$. This is also a contradiction because $\nu(G^*\setminus S')=\nu(G^*)$. It follows that $\size{S'}=k$. 

Without loss of generality, assume $S'=\set{v'_1,v'_2,\dots,v'_k}$. Let $S:=\set{v_1,v_2,\dots,v_k}$. We claim that every $v_i\in S$ is matched to some $b_j$ in $M$. For the purpose of contradiction, suppose there exists $v_i\in S$ such that $v_i$ is $M$-exposed. By the pigeonhole principle, there exists $j\in [n]$ such that $v_jv'_j\in G^*\setminus S'$ and $b_\ell v_j\in M$ for some $\ell\in [k]$. Then, $(v_i,b_\ell,v_j,v'_j)$ forms an augmenting path, which is a a contradiction. It is left to show that $S$ is an independent set. For the purpose of contradiction, suppose there exist $v_i,v_j\in S$ such that $v_iv_j\in E$. Let $b_pv_i,b_qv_j\in M$ for some $p,q\in [k]$. Then, $C_p\cup (b_p,v_i,v_j,b_q) \cup C_q$ forms an augmenting bi-cycle. By Theorem \ref{thm:stable}, $G^*\setminus S'$ is unstable, which is a contradiction. Thus, $S$ is an independent set and $\size{S}\geq k$.
\end{proof}

\section{Computing edge-stabilizers}
\label{sec:edge_stabilizer}

In contrast to the vertex-stabilizer problem, finding a minimum edge-stabilizer is computationally difficult. Bir\'{o} et al.~\cite{journals/tcs/BiroBGKP} proved that the problem is NP-hard on weighted graphs. We strengthen this result by showing the following hardness of approximation:   

\begin{theorem}
\label{thm:hardness_edge}
There is no constant factor approximation for the minimum edge-stabilizer problem unless $P=NP$.
\end{theorem}

\begin{proof}
We construct a gap-producing reduction from the independent set problem. Let $G=(V,E)$ and $k$ be an independent set instance where $V=\set{v_1,v_2,\dots,v_n}$. The independent set problem asks to determine whether $G$ has an independent set of size at least $k$. We may assume $2\leq k\leq n$. Let $\rho\geq 1$ be an integer. We construct the gadget graph $G^*=(V^*,E^*)$ as follows. For every edge $v_iv_j\in E$, replace it with $\rho k$ paths of length 3, i.e. $(v_i,u^\ell_{ij},u^\ell_{ji},v_j)$ for $\ell\in [\rho k]$. Assign weight 1 to every edge in the paths. For each $v_i\in V$, add a vertex $v'_i$ and the edge $v_iv'_i$ with weight 1. Next, create $k$ pairwise-disjoint copies of the complete graph on $2\rho k+1$ vertices. Denote each of them as $H_i$ and assign weight 4 to every edge in $H_i$. In addition, for every $H_i$, label one of the vertices as $b_i$. Finally, add the edge $b_iv_j$ for every $i\in [k]$ and $j\in [n]$ with weight 2. (See Figure \ref{fig:hardness_edge})

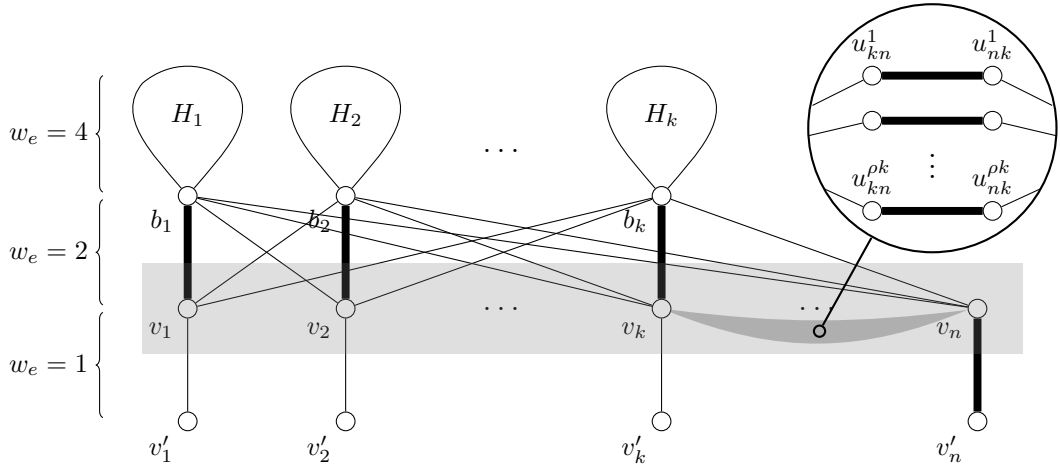
\begin{figure}[ht]
\centering
\def\dist{1.5}
\def\gap{0.4}
\begin{tikzpicture}[node distance=\dist cm, inner sep=2.5pt, minimum size=2.5pt, auto]
% Nodes & Dots
\node [vertex,label=below left:\small{$v_1$}] (v1) {};
\node [vertex,label=below left:\small{$v_2$}] (v2) [right of=v1, xshift=\gap*\dist cm] {};
\node [font=\large] (d1) [right of=v2, xshift=\gap*\dist cm] {$\dots$};
\node [vertex,label=below left:\small{$v_k$}] (v3) [right of=d1, xshift=\gap*\dist cm] {};
\node [font=\large] (d2) [right of=v3, xshift=\gap*\dist cm] {$\dots$};
\node [font=\large] (d3) [above of=d1, yshift=\gap*\dist cm] {$\dots$};
\node [vertex,label=below left:\small{$v_n$}] (v4) [right of=d2, xshift=\gap*\dist cm] {};

\node [vertex,label=below left:\small{$v'_1$}] (w1) [below of=v1] {};
\node [vertex,label=below left:\small{$v'_2$}] (w2) [below of=v2] {};
\node [vertex,label=below left:\small{$v'_k$}] (w3) [below of=v3] {};
\node [vertex,label=below left:\small{$v'_n$}] (w4) [below of=v4] {};

\node [vertex,label=below left:\small{$b_1$}] (b1) [above of=v1] {};
\node [vertex,label=below left:\small{$b_2$}] (b2) [above of=v2] {};
\node [vertex,label=below left:\small{$b_k$}] (b3) [above of=v3] {};

% The complete graphs
\draw (b1) to[out=130,in=-130] (-\gap*\dist,2*\dist) to[out=40,in=140] (\gap*\dist,2*\dist) to[out=-50,in=50] (b1);
\draw (b2) to[out=130,in=-130] (\dist,2*\dist) to[out=40,in=140] (\dist+2*\gap*\dist,2*\dist) to[out=-50,in=50] (b2);
\draw (b3) to[out=130,in=-130] (3*\dist+2*\gap*\dist,2*\dist) to[out=40,in=140] (3*\dist+4*\gap*\dist,2*\dist) to[out=-50,in=50] (b3);

\node [weight] (H1) [above of=b1, yshift=-0.7*\gap*\dist cm] {$H_1$};
\node [weight] (H2) [above of=b2, yshift=-0.7*\gap*\dist cm] {$H_2$};
\node [weight] (H3) [above of=b3, yshift=-0.7*\gap*\dist cm] {$H_k$};

% Edges
\path (v1) edge (w1);
\path (v2) edge (w2);
\path (v3) edge (w3);
\path [matched edge] (v4) -- (w4);
\path[fill=gray,opacity=0.5] (v3) to [bend right=7] (v4) to [bend left=22] (v3);

\path [matched edge] (b1) -- (v1);
\path (b1) edge (v2);
\path (b1) edge (v3);
\path (b1) edge (v4);

\path (b2) edge (v1);
\path [matched edge] (b2) -- (v2);
\path (b2) edge (v3);
\path (b2) edge (v4);

\path (b3) edge (v1);
\path (b3) edge (v2);
\path [matched edge] (b3) -- (v3);
\path (b3) edge (v4);

% Shaded Region
\path [fill=gray,opacity=0.25] (-\gap*\dist,-\gap*\dist)  -- (-\gap*\dist,\gap*\dist) -- (5*\dist+6*\gap*\dist,\gap*\dist) -- (5*\dist+6*\gap*\dist,-\gap*\dist) -- (-\gap*\dist,-\gap*\dist);

% Magnifying glass
\node [circle,draw,thick,inner sep=1.5pt] (small) [below of=d2,yshift=0.8*\dist cm] {};
\node [circle,draw,thick] (big) [minimum size=3.3cm,above of=v4, xshift=-\gap*\dist cm, yshift=1.5*\gap*\dist cm] {};
\path (small) edge [thick] (big);

\node [minimum size=0pt,inner sep=0pt] (s1) [left of=big,xshift=-0.1 cm,yshift=0.3 cm] {};
\node [minimum size=0pt,inner sep=0pt] (s2) [left of=big,xshift=-0.15cm,yshift=-0.1 cm] {};
\node [minimum size=0pt,inner sep=0pt] (s3) [left of=big,xshift=0.05 cm,yshift=-0.8 cm] {};

\node [minimum size=0pt,inner sep=0pt] (s4) [right of=big,xshift=0.1cm,yshift=0.3cm] {};
\node [minimum size=0pt,inner sep=0pt] (s5) [right of=big,xshift=0.15cm,yshift=-0.1cm] {};
\node [minimum size=0pt,inner sep=0pt] (s6) [right of=big,xshift=-0.05cm,yshift=-0.8cm] {};

\node [vertex,label=above:\small{$u^1_{kn}$}] (u1) [left of=big,xshift=0.7cm,yshift=0.7cm] {};
\node [vertex] (u2) [left of=big,xshift=0.7 cm,yshift=0.1 cm] {};
\node [vertex,label=above:\small{$u^{\rho k}_{kn}$}] (u3) [left of=big,xshift=0.7cm,yshift=-1.1 cm] {};

\node [vertex,label=above:\small{$u^1_{nk}$}] (u4) [right of=big,xshift=-0.7cm,yshift=0.7 cm] {};
\node [vertex] (u5) [right of=big,xshift=-0.7cm,yshift=0.1 cm] {};
\node [vertex, label=above:\small{$u^{\rho k}_{nk}$}] (u6) [right of=big,xshift=-0.7cm,yshift=-1.1 cm] {};
\node [font=\large] (d4) [below of=big, xshift=0 cm, yshift=1.1 cm] {$\vdots$};

\path [matched edge] (u1) -- (u4);
\path [matched edge] (u2) -- (u5);
\path [matched edge] (u3) -- (u6);
\path (s1) edge (u1);
\path (s2) edge (u2);
\path (s3) edge (u3);
\path (u4) edge (s4);
\path (u5) edge (s5);
\path (u6) edge (s6);

% Braces
\draw[decoration={brace,raise=15pt},decorate] (-\gap*\dist,\dist+0.05) -- node[weight,left=18pt] {$w_e=4$} (-\gap*\dist,2.1*\dist-0.05);
\draw[decoration={brace,raise=15pt},decorate] (-\gap*\dist,0.05) -- node[weight,left=18pt] {$w_e=2$} (-\gap*\dist,\dist-0.05);
\draw[decoration={brace,raise=15pt},decorate] (-\gap*\dist,-\dist+0.05) -- node[weight,left=18pt] {$w_e=1$} (-\gap*\dist,-0.05);
\end{tikzpicture}
\caption{The gadget graph $G^*$.}
\label{fig:hardness_edge}
\end{figure}

\begin{claim}\label{clm:gap1}
If $G$ has an independent set of size at least $k$, then $G^*$ has an edge-stabilizer of size at most $k$.
\end{claim}

\begin{proof}
Let $S$ be an independent set in $G$ where $\size{S}=k$. Without loss of generality, we may assume $S=\set{v_1,v_2,\dots,v_k}$. Let $F=\cup_{i=1}^k v_iv'_i$. We claim that $F$ is an edge-stabilizer of $G^*$. Let $M_i$ be a perfect matching in $H_i\setminus b_i$ for all $i\in [k]$. Let $\hat{M}:=\cup_{\ell=1}^{\rho k}\set{u^\ell_{ij}u^\ell_{ji}:v_iv_j\in E}$. Define the matching $M$ in $G^*\setminus F$ as 
\[M:=\hat{M}\cup M_1\cup\dots\cup M_k \cup \set{b_iv_i:1\leq i\leq k} \cup \set{v_iv'_i:k< i\leq n}\]
Note that $w(M)=(m+4k)\rho k + n + k$. In order to show that $G^*\setminus F$ is stable, it suffices to exhibit a fractional $w$-vertex cover of the same weight. Let $y\in \R^{\size{V^*}}_+$ be a vector defined by
\[y_v=\begin{cases}
   2, &\text{ if }v\in \cup_{i=1}^k V(H_i)\\
   1, &\text{ if }v\in \cup_{i=k+1}^n v_i \; \text{ or } \; v=u^\ell_{ij} \text{ where }i\leq k\\
   \frac{1}{2}, &\text{ if }v=u^\ell_{ij} \text{ where } i,j> k\\
   0, &\text{ otherwise}.
\end{cases}\]
It is easy to check that $y_u+y_v\geq w_{uv}$ for all $uv\in E^*$. Hence, $y$ is a fractional $w$-vertex cover in $G^*$. Since $S$ is an independent set, there are no edges of the form $u^\ell_{ij}u^\ell_{ji}$ where $i,j\leq k$. Then,
\[\1^{\top}y=2(2\rho k+1)k+n-k+m\rho k = 4\rho k^2 + n + k + m\rho k = (m+4k)\rho k + n + k= w(M)\]
which implies that $G^*\setminus F$ is stable.
\end{proof}

\begin{claim}\label{clm:gap2}
If $G$ does not have an independent set of size at least $k$, then every edge-stabilizer of $G^*$ has size at least $(\rho+1)k$.
\end{claim}

\begin{proof}
We prove the contrapositive. Assume $G^*$ has an edge-stabilizer $F$ such that $\size{F}<(\rho+1)k$. Let $M$ be a maximum-weight matching in $G^*\setminus F$. Let $c$ denote the number of edges removed from the complete graphs, i.e. $c:=\size{F\cap \cup_{i=1}^k E(H_i)}$. We first show that $c<2\rho k-1$. According to Ore's Theorem, we need to remove at least $2\rho k-1$ edges from $H_i$ in order to make it non-Hamiltonian. Let $\mathcal{H}:=\set{i:H_i\setminus F \mbox{ is Hamiltonian}}$. Then, $\size{\mathcal{H}}\geq k - \frac{c}{2\rho k-1}$. For every $i\in \mathcal{H}$, $b_iv_j\in M$ for some $j\in[n]$, otherwise $H_i\setminus F$ contains an augmenting blossom because it has an odd number of vertices. Thus, $v_jv'_j\subseteq F$, otherwise $H_i\setminus F\cup(b_i,v_j,v'_j)$ contains an augmenting flower. We have
\begin{align*}
c + \size{\mathcal{H}} &< (\rho+1)k \\
c + \pr{k-\frac{c}{2\rho k-1}} &< (\rho+1)k \\
c \pr{1 - \frac{1}{2\rho k-1}} &< \rho k \\
c \pr{\frac{2\rho k-2}{2\rho k-1}} &< \rho k \\
c &< (2\rho k-1)\pr{\frac{\rho k}{2\rho k-2}} \\
c &< 2\rho k-1
\end{align*}
Since $c<2\rho k-1$, $\size{\mathcal{H}}=k$. Without loss of generality, we may assume $b_iv_i\in M$ for every $i\in [k]$. Then, $\cup_{i=1}^k v_iv'_i\subseteq F$. We claim that $S=\set{v_1,v_2,\dots,v_k}$ is an independent set in $G$. For the purpose of contradiction, suppose there exists an edge $v_iv_j\in E$ for some $i,j\in[k]$. Let $\mathcal{P}_{ij}=\cup_{\ell=1}^{\rho k}(v_i,u^\ell_{ij},u^\ell_{ji},v_j)$ denote the set of paths between $v_i$ and $v_j$. Since $\size{F\cap \mathcal{P}_{ij}}<(\rho+1)k-k=\rho k$, at least one path $(v_i,u^t_{ij},u^t_{ji},v_j)\in\mathcal{P}_{ij}$ is present in $G^*\setminus F$. Observe that $u^t_{ij}u^t_{ji}\in M$, and $(H_i\setminus F) \cup (b_i,v_i,u^t_{ij},u^t_{ji},v_j,b_j)\cup (H_j\setminus F)$ contains an augmenting bi-cycle. Thus, by Theorem \ref{thm:stable} $G^*\setminus F$ is unstable, which is a contradiction.
\end{proof}

Now, suppose we have an $\alpha$-approximation to the minimum edge-stabilizer problem for some constant $\alpha\geq 1$. Set $\rho = \ceil{\alpha}$ and construct the gadget graph $G^*$ as shown above. Run this algorithm on $G^*$ and let $F$ be the returned edge-stabilizer. Let OPT be size of a minimum edge-stabilizer in $G^*$. If $G$ has an independent set of size at least $k$, then by Claim \ref{clm:gap1} we have $\text{OPT} \leq k$ and $\size{F}\leq \alpha\cdot \text{OPT} \leq \rho k$. On the other hand, if $G$ does not have an independent set of size at least $k$, then by Claim \ref{clm:gap2} we have $\text{OPT}\geq (\rho+1)k>\rho k$. This implies that $\size{F}>\rho k$. Therefore, we can use this algorithm to decide the independent set problem in polynomial time.
\end{proof}

In this section, we prove that Algorithm \ref{alg:vert} is an $O(\Delta)$-approximation algorithm for the minimum edge-stabilizer problem. We first need to establish a lower bound on the optimal solution. Next example shows that, differently from the unweighted case, $\gamma(G)$ is not a lower bound on the size of a minimum edge-stabilizer for weighted graphs. Let $G$ be the unstable graph depicted in Figure \ref{fig:bound_edge}. The unique maximum-weight matching is shown in the left, while the unique maximum-weight fractional matching is shown in the right. Gray edges have value $\frac{1}{2}$. Even though $\gamma(G)=2$, the edge with weight 0.5 is a minimum edge-stabilizer.

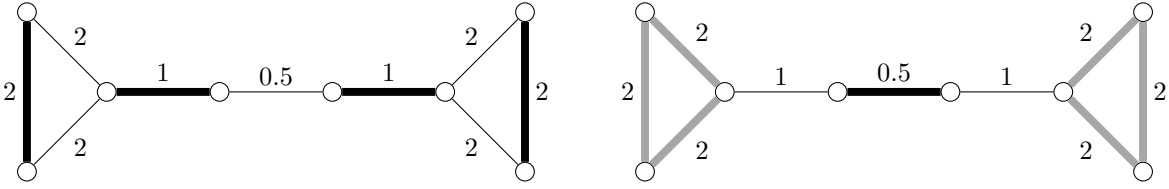
\begin{figure}[ht]
\def\dist{1.5}
\centering
\begin{minipage}{0.49\textwidth}
\centering
\begin{tikzpicture}[node distance=\dist cm, inner sep=2.5pt, minimum size=2.5pt, auto]
  % Nodes
  \node [vertex] (v1) {};
  \node [vertex] (v2) [above left of=v1] {};
  \node [vertex] (v3) [below left of=v1] {};
  \node [vertex] (v4) [right of=v1] {};
  \node [vertex] (v5) [right of=v4] {};
  \node [vertex] (v6) [right of=v5] {};
  \node [vertex] (v7) [above right of=v6] {};
  \node [vertex] (v8) [below right of=v6] {};

  % Edges
  \path [matched edge] (v2) -- node [weight, left] {$2$} (v3);
  \path (v1) edge node [weight, above right] {$2$} (v2);
  \path (v1) edge node [weight, below right] {$2$} (v3);
  \path [matched edge] (v1) -- node [weight, above] {$1$} (v4);
  \path (v4) edge node [weight, above] {$0.5$} (v5);
  \path [matched edge] (v5) -- node [weight, above] {$1$} (v6);
  \path (v6) edge node [weight, above left] {$2$} (v7);
  \path (v6) edge node [weight, below left] {$2$} (v8);
  \path [matched edge] (v7) -- node [weight, right] {$2$} (v8);
\end{tikzpicture}
\end{minipage}
\begin{minipage}{0.49\textwidth}
\centering
\begin{tikzpicture}[node distance=\dist cm, inner sep=2.5pt, minimum size=2.5pt, auto]
  % Nodes
  \node [vertex] (v1) {};
  \node [vertex] (v2) [above left of=v1] {};
  \node [vertex] (v3) [below left of=v1] {};
  \node [vertex] (v4) [right of=v1] {};
  \node [vertex] (v5) [right of=v4] {};
  \node [vertex] (v6) [right of=v5] {};
  \node [vertex] (v7) [above right of=v6] {};
  \node [vertex] (v8) [below right of=v6] {};

  % Edges
  \path [selected edge] (v2) -- node [weight, left] {$2$} (v3);
  \path [selected edge] (v1) -- node [weight, above right] {$2$} (v2);
  \path [selected edge] (v1) -- node [weight, below right] {$2$} (v3);
  \path (v1) edge node [weight, above] {$1$} (v4);
  \path [matched edge] (v4) -- node [weight, above] {$0.5$} (v5);
  \path (v5) edge node [weight, above] {$1$} (v6);
  \path [selected edge] (v6) -- node [weight, above left] {$2$} (v7);
  \path [selected edge] (v6) -- node [weight, below left] {$2$} (v8);
  \path [selected edge] (v7) -- node [weight, right] {$2$} (v8);
\end{tikzpicture}
\end{minipage}
\caption{An example showing that $\gamma(G)$ is not a lower bound.}
\label{fig:bound_edge}
\end{figure}

However, we can prove the following.

\begin{lemma}
\label{lem:gamma_edge}
For every edge $e\in E$, $\gamma(G\setminus e)\geq \gamma(G)-2$.
\end{lemma}

\begin{proof}
Let $x^*$ be a basic maximum-weight fractional matching in $G$ such that $\size{\mathscr{C}(x^*)}=\gamma(G)$. Let $y$ be a minimum fractional $w$-vertex cover in $G$. Pick an edge $ab\in E$.  If $x^*_{ab}=0$, then $\nu_f(G\setminus ab)=\nu_f(G)$ and $\gamma(G\setminus ab)= \gamma(G)$. So we may assume $x^*_{ab}\in \set{\frac{1}{2},1}$.  Let $G'$ be the graph obtained by replacing the edge $ab$ with edges $ab',b'a',a'b$ where $w_{ab'}=w_{b'a'}=w_{a'b}=w_{ab}$ and $a',b'\notin V$. Define the vectors $\hat{x}$ and $\hat{y}$ as
\[\hat{x}_e = \begin{cases}
  1-x^*_{ab}, &\text{ if }e=b'a',\\
  x^*_{ab}, &\text{ if }e\in \set{ab',a'b},\\
  x^*_e, &\text{ otherwise}.
\end{cases} \qquad
\hat{y}_v = \begin{cases}
  y_a, &\text{ if }v=a',\\
  y_b, &\text{ if }v=b',\\
  y_v, &\text{ otherwise}.
\end{cases}\]
Note that $\hat{x}$ is a basic fractional matching while $\hat{y}$ is a fractional $w$-vertex cover in $G'$. Furthermore, they satisfy complementary slackness conditions as $\hat{y}_u+\hat{y}_v=w_{uv}$ for all $uv\in\set{ab',b'a',a'b}$ and $\hat{x}(\delta(v))=1$ for all $v\in \set{a,a',b,b'}$. Hence, they form a primal-dual optimal pair. Since $\size{\mathscr{C}(\hat{x})}=\gamma(G)$, we have $\gamma(G')\leq \gamma(G)$. We claim that $\gamma(G')=\gamma(G)$. For the purpose of contradiction, suppose $\gamma(G')<\gamma(G)$. By Theorem \ref{thm:cycles}, $G'$ contains one of the following:

\smallskip
\emph{Structure (i):} a vertex $v\in V(C_i)$ for some odd cycle $C_i\in \mathscr{C}(\hat{x})$ such that $\hat{y}_v=0$.
If $v\in V$, then $v\in V(\mathscr{C}(x^*))$ and $y_v=0$. Otherwise, $v\in\set{a',b'}$ which implies $a,b\in V(\mathscr{C}(x^*))$ and $y_a=0$ or $y_b=0$. By Theorem \ref{thm:cycles}, we arrive at the contradiction $\gamma(G)<\size{\mathscr{C}(x^*)}$.

\smallskip
\emph{Structure (ii):} a tight $M(\hat{x})$-alternating path $P$ connecting two odd cycles $C_i,C_j\in \mathscr{C}(\hat{x})$.
If $P$ does not have any intermediate vertices from $\set{a,a',b,b'}$, then it is also a tight $M(x^*)$-alternating path in $G$ connecting two odd cycles from $\mathscr{C}(x^*)$. Otherwise, $ab',b'a',a'b\in E(P)$ and $C_i,C_j\in \mathscr{C}(x^*)$. Then, $(P\cup ab)\setminus\set{ab',b'a',a'b}$ is a tight $M(x^*)$-alternating path in $G$ connecting $C_i$ and $C_j$. By Theorem \ref{thm:cycles}, we arrive at the contradiction $\gamma(G)<\size{\mathscr{C}(x^*)}$.

\smallskip
\emph{Structure (iii):} a tight and valid $M(\hat{x})$-alternating path $P$ connecting an odd cycle $C_i\in \mathscr{C}(\hat{x})$ and a vertex $v\notin V(\mathscr{C}(\hat{x}))$ such that $\hat{y}_v=0$.
If $P$ does not have any intermediate vertices from $\set{a,a',b,b'}$, then $v\notin\set{a,a',b,b'}$. Hence, $P$ is also a tight and valid $M(x^*)$-alternating path in $G$ connecting an odd cycle from $\mathscr{C}(x^*)$ and $v$. If $ab',b'a',a'b\in E(P)$, then $v\notin\set{a',b'}$ and $C_i\in \mathscr{C}(x^*)$. Thus, $(P\cup ab)\setminus \set{ab',b'a',a'b}$ is a tight and valid $M(x^*)$-alternating path in $G$ connecting $C_i$ and $v$. If $a'b\in E(P)$ but $ab',b'a'\notin E(P)$, then $v=a'$ and $y_a=0$. So $P+ab-a'b$ is a tight and valid $M(x^*)$-alternating path in $G$ connecting $C_i$ and $a$. If $ab'\in E(P)$ but $b'a',a'b\notin E(P)$, then $v=b'$ and $y_b=0$. So $P+ab-ab'$ is a tight and valid $M(x^*)$-alternating path in $G$ connecting $C_i$ and $b$. By Theoreom \ref{thm:cycles}, we arrive at the contradiction $\gamma(G)<\size{\mathscr{C}(x^*)}$.

\smallskip
Thus, we have shown that $\gamma(G')=\gamma(G)$. Since $G'\setminus\set{a',b'}=G\setminus ab$, applying Lemma \ref{lem:gamma_vert} yields
\[\gamma(G\setminus ab)=\gamma\pr{G'\setminus\set{a',b'}}\geq \gamma(G')-2 = \gamma(G)-2.\]
\end{proof}

As a corollary to the above lemma, we obtain the following lower bound.

\begin{lemma}
For every edge-stabilizer $F$ of $G$, $\size{F}\geq \ceil{\frac{\gamma(G)}{2}}$.
\end{lemma}

Since Algorithm \ref{alg:vert} deletes $\gamma(G)$ vertices, at most $\gamma(G)\Delta$ edges are deleted, proving the following.

\begin{theorem}
There exists an efficient $O(\Delta)$-approximation algorithm for the minimum edge-stabilizer problem.
\end{theorem}

\section{Forcing an outcome}
\label{sec:additional}

Given a set of deals $M$, we here look at the problem of removing as few players as possible in order to make $M$ realizable as a stable outcome. This corresponds to finding a minimum vertex-stabilizer $S$ with the additional constraint that $M$ is a maximum-weight matching in $G\setminus S$. Note that this implicitly implies $S\cap V(M)=\emptyset$. A solution to this problem is called an \emph{$M$-vertex-stabilizer}. We would like to point out that the following variants of the problem, which are along the lines of Chandrasekaran et al. \cite{arxiv/ChandrasekaranG16} are polytime solvable: find a weight vector $w'$ such that $M$ is a stable outcome for $(G,w')$ so as to minimize $\|w-w'\|_1$ or $\|w-w'\|_\infty$ (or even the $\ell_p$ norm $\|w-w'\|_p$). This is an inverse-optimization problem that can be cast as an LP (or convex program for the $\ell_p$ norm) by exploiting complementary slackness. 

Ahmadian et al.~\cite{conf/ipco/AhmadianHS16} previously showed that the $M$-vertex-stabilizer problem is polynomial-time solvable on unweighted graphs when $M$ is a maximum matching in $G$. We prove that when $M$ is any arbitrary matching in $G$, the problem becomes hard:

\begin{theorem}\label{thm:hardness_mvert}
The $M$-vertex-stabilizer problem is NP-hard on unweighted graphs. Furthermore, no efficient $(2-\varepsilon)$-approximation algorithm exists for any $\varepsilon>0$ assuming UGC.
\end{theorem}

\begin{proof}
We give an approximation-preserving reduction from the vertex cover problem. Let $G=(V,E)$ be a vertex cover instance. For every edge $uv\in E$, replace it with an augmenting path of length three, i.e. $(u,u',v',v)$ where $u'v'\in M$. Denote the resulting (unweighted) graph as $G'=(V',E')$. We will show that every vertex cover in $G$ corresponds to an $M$-vertex-stabilizer in $G'$ and vice versa. This implies that the reduction is approximation-preserving and the inapproximability results for the vertex cover problem \cite{journals/am/DinurS05,journals/jcss/KhotR08} carry over to the problem of finding a minimum $M$-vertex-stabilizer.

Observe that $G'$ does not contain any alternating cycle or blossom. This implies that there is no augmenting cycle, flower or bi-cycle in $G'$. Let $S$ be a vertex cover of $G$. Then, $G'\setminus S$ has no augmenting path because every augmenting path in $G'$ corresponds to an edge in $G$. Thus, $M$ is a maximum matching in $G'\setminus S$ and hence $S$ is an $M$-vertex stabilizer. For the converse, suppose $S$ is an $M$-vertex stabilizer of $G'$. Note that $S\subseteq V$ as every vertex in $V'\setminus V$ is $M$-covered. Then, $G\setminus S$ has no edges because every edge in $G$ corresponds to an augmenting path in $G'$. It follows that $S$ is a vertex cover for $G$. This concludes the proof of inapproximability.
\end{proof}

On unweighted graphs, every instance of this problem admits a solution. However, this is not the case for weighted graphs. Consider an $M$-augmenting bi-cycle. It is unstable by Theorem \ref{thm:stable}, but does not have an $M$-vertex-stabilizer because every vertex is $M$-covered. In general, if the graph contains an $M$-augmenting path whose endpoints are $M$-covered, or an $M$-augmenting cycle, or an $M$-augmenting flower whose root is $M$-covered, or an $M$-augmenting bi-cycle, then it does not have an $M$-vertex-stabilizer. We would like to point out that recognizing an infeasible instance of the $M$-vertex-stabilizer problem can be done in polynomial time. In particular, we prove that:

\begin{theorem}\label{thm:mvert_stabilizer}
The $M$-vertex-stabilizer problem admits an efficient 2-approximation algorithm. Furthermore, if $M$ is a maximum-weight matching, then it can be solved in polynomial time.
\end{theorem}

We first sketch the main ideas. Given a weighted graph $G$ and a matching $M$, the algorithm searches for the structures which prevent $G$ from being stable or $M$ from being a maximum-weight matching. Among all such structures, the ones that can be tampered with are augmenting paths with at least one $M$-exposed endpoint and augmenting flowers with an $M$-exposed root. The algorithm then proceeds to delete these $M$-exposed vertices. If there exist augmenting paths whose endpoints are both $M$-exposed, the problem becomes hard because we do not know which endpoint is optimal to remove. In this case, the algorithm removes both endpoints, thus yielding a 2-approximation. Note that this last case cannot happen if $M$ is maximum-weight, and this explains why the problem becomes polynomial-time solvable. Kleinberg and Tardos \cite{conf/stoc/KleinbergT08} were the first to give a method of locating these structures. It involves solving a certain linear program using the dynamic programming algorithm of Aspvall and Shiloach \cite{journals/siamcomp/AspvallS80}. We use a slightly different algorithm for finding these structures, which in fact, will allow us to prove a strengthened version of Theorem \ref{thm:mvert_stabilizer} (Theorem \ref{thm:mvert_stabilizer_2}). Our algorithm relies on searching for augmenting walks of a certain length, via a slight modification of the dynamic programming algorithm given by Aspvall and Shiloach \cite{journals/siamcomp/AspvallS80}. Using this as a subroutine, we design an algorithm for the $M$-vertex-stabilizer problem.

\subsection{Finding augmenting walks}
The first ingredient is an algorithm to find augmenting walks of a certain length. We say that a valid $M$-alternating $uv$-walk $P$ of length at most $k$ is $\emph{optimal}$ if for any other valid $M$-alternating $uv$-walk $P'$ of length at most $k$, we have $w(P\setminus M)-w(P\cap M)\geq w(P'\setminus M) - w(P'\cap M)$. Given a source vertex $s$ and an integer $k\in \Z_+$, the algorithm searches for optimal valid alternating $sv$-walks of length at most $k$ for all $v\in V$. The significance of optimality is as follows. Let $P$ be an optimal valid alternating $sv$-walk returned by the algorithm. If $P$ is augmenting, then we have found an augmenting $sv$-walk of length at most $k$. Otherwise, we can conclude that there are no augmenting $sv$-walks of length at most $k$ by the optimality of $P$.

The inner workings of our algorithm is similar to the Grapevine algorithm given by Aspvall and Shiloach \cite{journals/siamcomp/AspvallS80}. In their paper, the Grapevine algorithm is used as a subroutine to solve linear systems of the form $Ax\leq b$, where each constraint contains at most 2 variables. They first constructed an auxiliary graph to model the relationship between variables and constraints. The Grapevine algorithm is then run on this auxiliary graph to compute lower and upper bounds of each variable.

We now give an overview of our algorithm. For every vertex $v\in V$, we define two variables $y_1(v)$ and $y_2(v)$. In iteration $i$, if $v$ is $M$-exposed, we would like $y_1(v)$ to represent the quantity $w(P\setminus M)-w(P\cap M)$, where $P$ is an optimal valid alternating $sv$-walk of length at most $i$. On the other hand, if $v$ is $M$-covered, we would like $y_2(v)$ to represent the quantity $w(P\setminus M)-w(P\cap M)$, where $P$ is an optimal valid alternating $sv$-walk of length at most $i$. We first initialize $y_1(s)=0$ and $y_2(s)=-\infty$ if $s$ is $M$-covered, and $y_1(s)=y_2(s)=0$ if $s$ is $M$-exposed. For every other vertex $v\neq s$, we set $y_1(v)=y_2(v)=-\infty$. At every iteration, each vertex $v$ determines whether it could increase its $y_1(v)$ value by replacing it with $y_2(u)+w_{uv}$ for some $uv\in E\setminus M$, and similarly, whether it could increase its $y_2(v)$ value by replacing it with $y_1(u)-w_{uv}$ for some $uv\in M$. In a way, this is analogous to the Bellman-Ford algorithm for computing shortest paths. The main difference is that we are maintaining two variables for each vertex, instead of one. This ensures the walk we obtain is $M$-alternating and valid.  Also, notice that we are adding the weights of unmatched edges and subtracting the weights of matched edges. This will give us our desired quantity $w(P\setminus M)-w(P\cap M)$.

\begin{algorithm}[H]
\SetKwInOut{Input}{Input}
Initialize vectors $y_1,y_2,z_1,z_2\in \R^n$\;
\uIf{$s$ is $M$-covered}{
  $y_1(s) \leftarrow 0$\;
  $y_2(s) \leftarrow -\infty$\;
}
\Else{
  $y_1(s) \leftarrow 0$\;
  $y_2(s) \leftarrow 0$\;
}
\ForEach{vertex $v\neq s$}{
  $y_1(v)\leftarrow -\infty$\;
  $y_2(v)\leftarrow -\infty$\;
}

\For{$i=1$ to $k$}{
  \ForEach{vertex $v$}{
    $\displaystyle z_1(v)\leftarrow\max_{u:uv\in E\setminus M}\set{y_2(u)+w_{uv}}$\;
    $\displaystyle z_2(v)\leftarrow\max_{u:uv\in M}\set{y_1(u)-w_{uv}}$\;
  }
  \ForEach{vertex $v$}{
    $y_1(v)\leftarrow\max(y_1(v),z_1(v))$\;
    $y_2(v)\leftarrow\max(y_2(v),z_2(v))$\;
  }
}

\Return{$y_1,y_2$}
\caption{Optimal valid $M$-alternating $sv$-walks of length at most $k$}
\label{alg:optwalk}
\end{algorithm}

In the algorithm, we take the maximum of the empty set to be $-\infty$. For the analysis, we will use $y^i_1(v)$ and $y^i_2(v)$ to denote the value of $y_1(v)$ and $y_2(v)$ respectively at iteration $i$ for all $i< k$. Note that $y^0_1(v)$ and $y^0_2(v)$ refer to the initial value received by $y_1(v)$ and $y_2(v)$ before the main ``for'' loop. The following lemmas verify our intuition.

\begin{lemma}\label{lem:covered}
Let $v$ be an $M$-covered vertex. If there is no valid $M$-alternating $sv$-walk of length at most $k$, then $y_2(v)=-\infty$. Otherwise, there exists an optimal valid $M$-alternating $sv$-walk $P$ of length at most $k$, and $y_2(v)=w(P\setminus M)-w(P\cap M)$.
\end{lemma}

\begin{proof}
We start by proving the contrapositive of the first statement. Let $v$ be an $M$-covered vertex where $y_2(v)$ is finite. We proceed by induction on $k$. We look at two base cases. When $k=1$, $y_2(v)=y^0_1(s)-w_{sv}$ where $sv\in M$. So $(s,v)$ is our desired walk. When $k=2$, if $y_2(v)$ was updated in iteration 1, then this reduces to the previous case. Otherwise, $y_2(v)=y^0_2(s)+w_{su}-w_{uv}$ for some $uv\in M$. Since $y^0_2(s)$ is finite, $s$ is $M$-exposed and $(s,u,v)$ is our desired walk. For the inductive hypothesis, assume the statement holds for some $k\geq 2$. Consider the case $k+1$. Let $j$ be the last iteration in which $y_2(v)$ was updated, i.e. $y_2(v)=y_1^{j-1}(u)-w_{uv}$ for some $uv\in M$. We may assume $j>2$, otherwise we are back at the base cases. Since the update of $y_2(v)$ was triggered by the update of $y_1(u)$, we know that $y_1(u)$ was updated at iteration $j-1$, i.e. $y_1^{j-1}(u)=y^{j-2}_2(t)+w_{tu}$ for some $tu\in E\setminus M$. Similarly, since the update of $y_1(u)$ was triggered by the update of $y_2(t)$,  we know that $y_2(t)$ was updated at iteration $j-2>0$. This implies that $t$ is $M$-covered. As $y^{j-2}_2(t)$ is finite and $j-2<k$, by the inductive hypothesis there exists a valid $M$-alternating $st$-walk of length at most $j-2$. Appending $(t,u,v)$ to this walk yields a valid $M$-alternating $sv$-walk of length at most $j\leq k+1$.

Next, we prove the second statement. Let $v$ be an $M$-covered vertex where a valid $M$-alternating $sv$-walk of length at most $k$ exists. Since the number of such walks is finite, there exists an optimal one. Among all such optimal walks, let $P$ be the shortest one in terms of number of edges. We proceed by induction on $k$. We look at two base cases. When $k=1$, $P=(s,v)$ and $y_2(v)=-w_{sv}=w(P\setminus M)-w(P\cap M)$. When $k=2$, if $\size{E(P)}=1$ then this reduces to the previous case. Otherwise, $P=(s,u,v)$ for some $uv\in M$ and $y_2(v)=w_{su}-w_{uv}=w(P\setminus M)-w(P\cap M)$. For the inductive hypothesis, assume the statement holds for some $k\geq 2$. Consider the case $k+1$. We may assume $P$ has length exactly $k+1$, otherwise by the inductive hypothesis we are done. Denote $P=(v_0,v_1,\dots,v_{k+1})$ where $v_0=s$ and $v_{k+1}=v$. Then, $P'=(v_0,v_1,\dots,v_{k-1})$ is an optimal valid $M$-alternating $sv_{k-1}$-walk of length at most $k-1$. Since $v_{k-1}$ is $M$-covered, by the inductive hypothesis we have $y^{k-1}_2(v_{k-1})=w(P'\setminus M)-w(P'\cap M)$. We also know that $y^{k-1}_2(v_{k-1})+w_{v_{k-1}v_k}\geq y^{k-1}_2(u)+w_{uv_k}$ for all $uv_k\in E\setminus M$ because $P$ is optimal. We claim that $y^{k-1}_2(v_{k-1})+w_{v_{k-1}v_k}>y^{k-1}_1(v_k)$. For the purpose of contradiction, suppose otherwise. Then, 
\begin{align*}
  y_2^k(v) &\geq y_1^{k-1}(v_k)-w_{v_kv} \\
           &\geq y^{k-1}_2(v_{k-1})+w_{v_{k-1}v_k}-w_{v_kv} \\
           &= w(P'\setminus M) - w(P'\cap M) +w_{v_{k-1}v_k}-w_{v_kv} \\
           &= w(P\setminus M) - w(P\cap M).
\end{align*}
Since $y_2^k(v)$ is finite, by the first part of the lemma there exists a valid $M$-alternating $sv$-walk of length at most $k$. So by the inductive hypothesis, $y_2^k(v)=w(Q\setminus M)-w(Q\cap M)$ where $Q$ is an optimal valid $M$-alternating $sv$-walk of length at most $k$. Note that $w(Q\setminus M)-w(Q\cap M)=w(P\setminus M)-w(P\cap M)$ because $P$ is optimal. However, $Q$ is shorter than $P$, which is a contradiction. Thus, we obtain $y_1^k(v_k)=y_2^{k-1}(v_{k-1})+w_{v_{k-1}v_k}$ and 
\[y_2(v) = y_1^k(v_k)-w_{v_kv} = y_2^{k-1}(v_{k-1})+w_{v_{k-1}v_k}-w_{v_kv} = w(P\setminus M) - w(P\cap M).\]
\end{proof}

\begin{lemma}\label{lem:exposed}
Let $v$ be an $M$-exposed vertex. If there is no valid $M$-alternating $sv$-walk of length at most $k$, then $y_1(v)=-\infty$. Otherwise, there exists an optimal valid $M$-alternating $sv$-walk $P$ of length at most $k$, and $y_1(v)=w(P\setminus M)-w(P\cap M)$.
\end{lemma}

\begin{proof}
We start by proving the contrapositive of the first statement. Let $v$ be an $M$-exposed vertex where $y_1(v)$ is finite. We proceed by induction on $k$. We look at two base cases. When $k=0$, we have $v=s$ and the empty path $(s)$ is our desired walk. When $k=1$, if $y_1(v)$ was never updated, then $v=s$ and so this reduces to the previous case. Otherwise, $y_1(v)=y^0_2(s)+w_{sv}$. Since $y^0_2(s)$ is finite, $s$ is $M$-exposed and $(s,v)$ is our desired walk. For the inductive hypothesis, assume the statement holds for some $k\geq 1$. Consider the case $k+1$. Let $j$ be the last iteration in which $y_1(v)$ was updated, i.e. $y_1(v)=y_2^{j-1}(u)+w_{uv}$ for some $uv\in E\setminus M$. We may assume $j>1$, otherwise we are back at the base cases. Since the update of $y_1(v)$ was triggered by the update of $y_2(u)$, we know that $y_2(u)$ was updated at iteration $j-1>0$. This implies that $u$ is $M$-covered. As $y^{j-1}_2(u)$ is finite and $j-1\leq k$, by Lemma \ref{lem:covered} there exists a valid $M$-alternating $su$-walk of length at most $j-1$. Appending $(u,v)$ to this walk yields a valid $M$-alternating $sv$-walk of length at most $j\leq k+1$.

Next, we prove the second statement. Let $v$ be an $M$-exposed vertex where a valid $M$-alternating $sv$-walk of length at most $k$ exists. Since the number of such walks is finite, there exists an optimal one. Among all such optimal walks, let $P$ be the shortest one in terms of number of edges. We proceed by induction on $k$. We look at two base cases. When $k=0$, $P=(s)$ and $y_1(v)=0=w(P\setminus M)-w(P\cap M)$. When $k=1$, if $\size{E(P)}=0$ then this reduces to the previous case. Otherwise, $P=(s,v)$ and $y_1(v)=w_{sv}=w(P\setminus M)-w(P\cap M)$. For the inductive hypothesis, assume the statement holds for some $k\geq 1$. Consider the case $k+1$. We may assume $P$ has length exactly $k+1$, otherwise by the inductive hypothesis we are done. Denote $P=(v_0,v_1,\dots,v_{k+1})$ where $v_0=s$ and $v_{k+1}=v$. Then, $P'=(v_0,v_1,\dots,v_k)$ is an optimal valid $M$-alternating $sv_k$-walk of length at most $k$. Since $v_k$ is $M$-covered, by Lemma \ref{lem:covered} we have $y^k_2(v_k)=w(P'\setminus M)+w(P'\cap M)$. We also know $y^k_2(v_k)+w_{v_kv}\geq y^k_2(u)+w_{uv}$ for all $uv\in E\setminus M$ because $P$ is optimal. We claim that $y^k_2(v_k)+w_{v_kv}>y^k_1(v)$. For the purpose of contradiction, suppose otherwise. Then 
\[y^k_1(v)\geq y^k_2(v_k)+w_{v_kv}=w(P'\setminus M)-w(P'\cap M) + w_{v_kv} = w(P\setminus M)-w(P\cap M).\]
Since $y_1^k(v)$ is finite, by the first part of the lemma there exists a valid $M$-alternating $sv$-walk of length at most $k$. So by the inductive hypothesis, $y_1^k(v)=w(Q\setminus M)-w(Q\cap M)$ where $Q$ is an optimal valid $M$-alternating $sv$-walk of length at most $k$. Note that $w(Q\setminus M)-w(Q\cap M)=w(P\setminus M)-w(P\cap M)$ because $P$ is optimal. However, $Q$ is shorter than $P$, which is a contradiction. Thus, we obtain $y_1(v)=y_2^k(v_k)+w_{v_kv}=w(P\setminus M)-w(P\cap M)$.
\end{proof}

\subsection{The algorithm}
The reason we look for augmenting walks is because of the following:

\begin{lemma}\label{lem:augwalk}
An augmenting $uv$-walk contains an augmenting $uv$-path, an augmenting cycle, an augmenting flower rooted at $u$ or $v$, or an augmenting bi-cycle. 
\end{lemma}

\begin{proof}
We first prove the following claim:

\begin{claim}\label{clm:decompose}
  If $P$ is an alternating walk, then it can be decomposed into $P=P_1P_2\dots P_\ell$ such that: 
  \begin{enumerate}[noitemsep,topsep=0pt]
    \item[(i)] Every $P_i$ is an alternating path, an alternating cycle or a blossom.
    \item[(ii)] There is no $i$ such that $P_i$ and $P_{i+1}$ are both alternating paths or blossoms.
  \end{enumerate}
\end{claim}

\begin{proof}
Let $P=(v_1,v_2,\dots,v_t)$ be an $M$-alternating walk. We proceed by induction on $t$. For the base case $t=2$, $P$ is an alternating path of length 1 as there are no loops in $G$. Suppose the lemma is true for $t\leq k$ for some $k\geq 2$. Consider the case $t=k+1$. We may assume $P$ is not simple. Let $j$ be the smallest index such that $v_j=v_i$ for some $i<j$. Decompose $P$ into $P_1=(v_1,v_2,\dots,v_i)$, $P_2=(v_i,v_{i+1},\dots,v_j)$ and $P_3=(v_j,v_{j+1},\dots,v_t)$. $P_1$ is a (possibly empty) alternating path while $P_2$ is an alternating cycle or a blossom. Since $P_3$ is an $M$-alternating walk with fewer edges, by the inductive hypothesis it can be decomposed into $P_3=P'_1P'_2\dots,P'_\ell$ where every $P'_i$ is an alternating path, an alternating cycle or a blossom. Moreover, there are no consecutive paths or blossoms in this decomposition. Note that $P'_1$ is not a blossom because $P_3$ starts with an edge in $M$. Thus, $P=P_1P_2P'_1P'_2\dots P'_\ell$ is our desired decomposition.
\end{proof}

Let $P$ be an $M$-augmenting $uv$-walk. Using Claim \ref{clm:decompose}, decompose $P$ into $P=P_1P_2\dots P_k$. If $P_i$ is an augmenting cycle for some $i\in [k]$, then we are done. So we may assume that every alternating cycle in the decomposition is not augmenting. Note that $P_k$ is not an alternating cycle, otherwise $P$ is not valid because it ends with an unmatched edge whose endpoints are $M$-covered. Let $P'$ be the alternating $uv$-walk obtained by dropping all the alternating cycles in the decomposition. It is easy to see that $P'$ is still augmenting. Repeat this process until we are left with an augmenting $uv$-walk $P^*=P^*_1P^*_2\dots P^*_\ell$ such that every $P_i$ is an alternating path or a blossom. 

If $P^*_1$ is a blossom, then $u$ is $M$-exposed because the first edge of $P^*$ is not in $M$. Similarly, if $P^*_\ell$ is a blossom, then $v$ is $M$-exposed because the last edge of $P^*$ is not in $M$. In both cases, since $P^*$ does not have any $M$-exposed intermediate vertices, we get $u=v$ and $\ell=1$. This implies that $P^*$ is an augmenting blossom with base $u$, which is trivially an augmenting flower with root $u$. Thus, we may assume $P_1$ and $P_\ell$ are alternating paths. If $\ell=1$, then $P^*$ is an augmenting $uv$-path. Otherwise, from Claim \ref{clm:decompose} we know that $P_i$ is an alternating path for all odd $i$ while $P_i$ is a blossom for all even $i$. Observe that $P^*_1\cup P^*_2$ and $P^*_{\ell-1}\cup P^*_\ell$ form flowers rooted at $u$ and $v$ respectively, where the former is simple while the latter might not be simple. Moreover, $P^*_{2i}\cup P^*_{2i+1} \cup P^*_{2i+2}$ form bi-cycles for all $i\in \br{\frac{\ell-3}{2}}$. Since $2w(P^*\setminus M)>2w(P^*\cap M)$ and
\[2w(P^*)=2w(P^*_1)+w(P^*_2) + \sum_{i=1}^{(\ell-3)/2}\pr{w(P^*_{2i})+2w(P^*_{2i+1})+w(P^*_{2i+2})} + w(P^*_{\ell-1}) + 2w(P^*_\ell),\]
at least one of them is augmenting.
\end{proof}

We are now ready to present the algorithm for the $M$-vertex-stabilizer problem:

\begin{algorithm}[H]
Initialize $S\leftarrow\emptyset$ \;
\ForEach{$M$-exposed vertex $u$}{
  Search for $M$-augmenting $uv$-walks of length at most $3n$ using Algorithm \ref{alg:optwalk}\;
  \uIf{$\exists$ an $M$-augmenting $uu$-walk or $uv$-walk for some $M$-covered vertex $v$}{
    $S\leftarrow S\cup\set{u}$\;
    $G\leftarrow G\setminus u$\;
  }
}
\ForEach{$M$-exposed vertex $u$}{
  Search for $M$-augmenting $uv$-walks of length at most $n$ using Algorithm \ref{alg:optwalk}\;
  \uIf{$\exists$ an $M$-augmenting $uv$-walk for some $M$-exposed vertex $v$}{
    $S\leftarrow S\cup\set{u,v}$\;
    $G\leftarrow G\setminus\set{u,v}$\;
  }
}
\uIf{$w(M)<\nu_f(G)$}{
    \Return ``INFEASIBLE'' \;
  }
  \uElse{
    \Return $S$\;

  }

\caption{$M$-vertex-stabilizer}
\label{alg:mvert}
\end{algorithm}

Let $S_1$ denote the set of $M$-exposed vertices in $G$ which are roots of augmenting flowers or endpoints of augmenting paths whose other endpoint is $M$-covered. Note that given a feasible instance, every $M$-vertex-stabilizer contains $S_1$. We prove a stronger statement than Theorem \ref{thm:mvert_stabilizer}:

\begin{theorem}\label{thm:mvert_stabilizer_2}
The $M$-vertex-stabilizer problem admits an efficient 2-approximation algorithm. Furthermore, if $M$ is a maximum-weight matching in $G\setminus S_1$, then it is polynomial-time solvable.
\end{theorem}

\begin{proof}
Let $G$ be the input graph and $M$ be a matching in $G$. Let $R$ be the set of $M$-exposed vertices in $G$, and let $R'$ be any subset of $R$. If $G$ contains an augmenting path whose endpoints are $M$-covered or an augmenting cycle, then it is also present in $G\setminus R'$. Since $M$ is not a maximum-weight matching in $G\setminus R'$, $R'$ is not an $M$-vertex-stabilizer. Similarly, if $G$ contains an augmenting flower whose root is $M$-covered or an augmenting bi-cycle, then it is also present in $G\setminus R'$. By Theorem \ref{thm:stable}, $G\setminus R'$ is not stable. In these two cases, there is no $M$-vertex-stabilizer. Since $S\subseteq R$ and $w(M)<v_f(G\setminus S)$, the algorithm will return ``INFEASIBLE''. 

Thus, we may assume $G$ does not contain any of the aforementioned structures. The only structure which can make $G$ unstable is an augmenting flower whose root is $M$-exposed. In addition, the only structure which can prevent $M$ from being a maximum-weight matching in $G$ is an augmenting path with at least one $M$-exposed endpoint. 

\begin{claim}
Let $u$ be an $M$-exposed vertex. Then, $u$ is the root of an augmenting flower if and only if there exists an augmenting $uu$-walk of length at most $3n$.
\end{claim}

\begin{proof}
Let $C\cup P$ be a augmenting flower rooted at $u$ where $C=(v_1,v_2,\dots,v_j,v_1)$ is the blossom and $P=(u_1,u_2,\dots,u_k)$ is the valid alternating path. Assume $u_1=v_1$ and $u_k=u$. Let $P^{-1}=(u_k,u_{k-1},\dots,u_1)$ denote the reverse of path $P$. Then, $Q=P^{-1}CP$ is a valid alternating $uu$-walk, and its length is at most $3n$. Moreover, since
\[w(Q\setminus M) =  w(C\setminus M) + 2w(P\setminus M) > w(C\cap M) + 2w(P\cap M) = w(Q\cap M),\]
it is augmenting. For the converse, let $P$ be an augmenting $uu$-walk of length at most $3n$. By Lemma \ref{lem:augwalk}, $P$ contains an augmenting flower rooted at $u$.
\end{proof}

\begin{claim}
Let $u$ be an $M$-exposed vertex and $v$ be an $M$-covered vertex. If there is no augmenting flower rooted at $u$, then there exists an augmenting $uv$-path if and only if there exists an augmenting $uv$-walk of length at most $3n$.
\end{claim}

\begin{proof}
A $uv$-path is trivially a $uv$-walk. For the converse, let $P$ be an augmenting $uv$-walk of length at most $3n$. By Lemma \ref{lem:augwalk}, $P$ contains an augmenting $uv$-path.
\end{proof}

By the two claims above, the set of vertices collected in the first ``for'' loop of the algorithm is exactly $S_1$. Let $S^*$ be a minimum $M$-vertex-stabilizer. Then, $S_1\subseteq S^*$. Now, the only structure which can prevent $M$ from being a maximum-weight matching in $G\setminus S_1$ is an augmenting path whose endpoints are both $M$-exposed.

\begin{claim}
Let $u$ and $v$ be $M$-exposed vertices. There exists an augmenting $uv$-path in $G\setminus S_1$ if and only if there exists an augmenting $uv$-walk of length at most $n$ in $G\setminus S_1$.
\end{claim}

\begin{proof}
A $uv$-path is trivially a $uv$-walk. For the converse, let $P$ be an augmenting $uv$-walk of length at most $n$. By Lemma \ref{lem:augwalk}, $P$ contains an augmenting $uv$-path.
\end{proof}

Let $S_2$ be the set of vertices collected in the second ``for'' loop of the algorithm. At every iteration, a pair of vertices were added to $S_2$ because they are the endpoints of an augmenting path. Note that at least one of them is in $S^*$, otherwise this augmenting path is present in $G\setminus S^*$. Thus, we have $\size{S^*}\geq \size{S_1}+\frac{1}{2}\size{S_2}\geq \frac{1}{2}\size{S}$. The matching $M$ is maximum-weight in $G\setminus(S_1\cup S_2)$ because there are no augmenting paths or cycles. Moreover, $G\setminus(S_1\cup S_2)$ is stable because it does not contain any augmenting flowers or bi-cycles. Thus, $S=S_1\cup S_2$ is an $M$-vertex-stabilizer. Finally, if $M$ is a maximum-weight matching in $G\setminus S_1$, then $S_2=\emptyset$. We get $\size{S}=\size{S_1}\leq \size{S^*}$ implying that $S$ is optimal.
\end{proof}

\bibliographystyle{abbrv}
\bibliography{References}

\section*{Appendix I: Omitted examples}
\subsection*{Minimum stabilizers do not preserve $\nu(G)$}
Here we demonstrate that there are graphs where the removal of any minimum edge- or vertex-stabilizer does not preserve the value of a maximum-weight matching. We first look at edge-stabilizers. Let $G$ denote the graph in Figure \ref{fig:preserve_edge}. It is not stable because $\nu(G)=8<9=\nu_f(G)$, where the maximum-weight fractional matching is given by
\[x_e = \begin{cases} \frac12, &\text{ if } e\in\set{pq,qr,pr}\\ 1, &\text{ if } e=st\\ 0, &\text{ otherwise}. \end{cases}\]

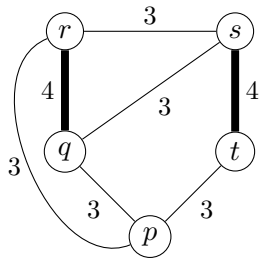
\begin{figure}[ht]
\centering
\def\dist{1.6}
\begin{tikzpicture}[node distance=\dist cm, inner sep=2.5pt, minimum size=2.5pt, auto]
  % Nodes
  \node [vertex] (v1) {$p$};
  \node [vertex] (v2) [above left of=v1] {$q$};
  \node [vertex] (v3) [above right of=v1] {$t$};
  \node [vertex] (v4) [above of=v2] {$r$};
  \node [vertex] (v5) [above of=v3] {$s$};

  % Edges
  \path [matched edge] (v2) -- node [weight, left] {$4$} (v4);
  \path [matched edge] (v3) -- node [weight, right] {$4$} (v5);
  \path (v1) edge node [weight, below left] {$3$} (v2);  
  \path (v1) edge node [weight, below right] {$3$} (v3);
  \path (v4) edge node [weight, above] {$3$} (v5);
  \path (v2) edge node [weight, below right] {$3$} (v5);
  \path (v1) edge[bend left=90] node [weight, left] {$3$} (v4);
\end{tikzpicture}
\caption{An example showing that the removal of any minimum edge-stabilizer does not preserve $\nu(G)$. Bold edges indicate the maximum-weight matching.}
\label{fig:preserve_edge}
\end{figure}

The minimum edge-stabilizer is $\set{qr}$. Observe that $M=\set{pq,st}$ is a maximum-weight matching in $G\setminus\set{qr}$ of weight 7. We can verify the stability of $G\setminus\set{qr}$ by constructing a fractional $w$-vertex cover $y$ of the same weight:
\[y_v = \begin{cases} 3, &\text{ if }v\in \set{p,s}\\ 1, &\text{ if }v=t\\ 0, &\text{ otherwise}.\end{cases}\]
It is left to show that if we delete any edge other than $qr$, the graph is still unstable. If edge $rs$, $qs$ or $pt$ is removed, the maximum-weight fractional matching remains the same. If edge $pr$ is removed, we get $\nu_f(G\setminus\set{pr})=8.5$ by assigning $x_e=1/2$ for all $e\in \set{pq,qr,rs,st,pt}$. If edge $pq$ is removed, we get $\nu_f(G\setminus\set{pq})=8.5$ by assigning $x_e=1/2$ for all $e\in \set{pr,rq,qs,st,pt}$. If edge $st$ is removed, we get $\nu(G\setminus\set{st})=7$ and $\nu_f(G\setminus\set{st})=8$ where the maximum-weight matching is $\set{qr,pt}$ and the maximum-weight fractional matching is
\[x_e = \begin{cases} \frac12, &\text{ if } e\in\set{qr,rs,qs}\\ 1, &\text{ if } e=pt\\ 0, &\text{ otherwise}. \end{cases}\]

The same negative result also holds for vertex-stabilizers. Consider the graph given in Figure \ref{fig:preserve_vert}. It is not stable because $\nu(G)=5<6=\nu_f(G)$, where the maximum-weight fractional matching is given by
\[y_v = \begin{cases} \frac{1}{2}, &\text{ if }e\in \set{pq,qr,pr}\\ 0, &\text{ otherwise}.\end{cases}\]
The minimum vertex-stabilizers of this graph are $\set{p},\set{q}$ and $\set{r}$. However, $\nu(G\setminus p)=\nu(G\setminus q)=\nu(G\setminus r)=4$.

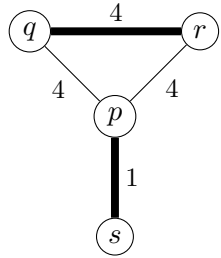
\begin{figure}[ht]
\centering
\def\dist{1.6}
\begin{tikzpicture}[node distance=\dist cm, inner sep=2.5pt, minimum size=2.5pt, auto]
  % Nodes
  \node [vertex] (v1) {$p$};
  \node [vertex] (v2) [above left of=v1] {$q$};
  \node [vertex] (v3) [above right of=v1] {$r$};
  \node [vertex] (v4) [below of=v1] {$s$};

  % Edges
  \path [matched edge] (v2) -- node [weight, above] {$4$} (v3);
  \path (v1) edge node [weight, below left] {$4$} (v2);
  \path [matched edge] (v1) -- node [weight, right] {$1$} (v4);  
  \path (v1) edge node [weight, below right] {$4$} (v3);
\end{tikzpicture}
\caption{An example showing that the removal of any minimum vertex-stabilizer does not preserve $\nu(G)$. Bold edges indicate the maximum-weight matching. }
\label{fig:preserve_vert}
\end{figure}

\end{document}